\documentclass[preprint,12pt]{elsarticle}
\usepackage{etoolbox}
\usepackage{tabularray}
\makeatletter
\def\ps@pprintTitle{%
	\let\@oddhead\@empty
	\let\@evenhead\@empty
	\def\@oddfoot{\reset@font\hfil\thepage\hfil}
	\let\@evenfoot\@oddfoot
}
\makeatother

\usepackage{bm}
\usepackage{amsmath, amsthm, amssymb}
\usepackage{amsfonts}
\usepackage{algorithm}
\usepackage{comment}
\usepackage{longtable}
\usepackage{caption}
\usepackage{booktabs}
\usepackage[colorlinks=true]{hyperref}
\usepackage{lscape}
\usepackage[utf8]{inputenc} % set input encoding (not needed with XeLaTeX)
\usepackage{amsthm}

\usepackage{enumerate}
\usepackage{mathrsfs}
\usepackage{placeins}
\usepackage{graphicx}
\usepackage{amsfonts}
\usepackage{verbatim}
\usepackage{amssymb}
\usepackage{amsthm}
\usepackage{multirow}
\usepackage{multicol}
\newtheorem{thm}{Theorem}[section]
\newtheorem{lem}{Lemma}[section]

\theoremstyle{definition}
\newtheorem{defn}{Definition}[section]
\newtheorem{ex}{Example}[section]
\newtheorem{res}{Result}[section]
\usepackage{algpseudocode}

\usepackage{tikz}
\usetikzlibrary{shapes.geometric, arrows}
\tikzstyle{startstop} = [rectangle, rounded corners, 
minimum width=10cm, 
minimum height=1cm,
text centered, 
draw=black, 
fill = blue!10]

\tikzstyle{io} = [trapezium, 
trapezium stretches=true, % A later addition
trapezium left angle=70, 
trapezium right angle=110, 
minimum width=3cm, 
minimum height=1cm, text centered, 
draw=black, fill = blue!10]

\tikzstyle{process} = [rectangle, 
minimum width=7cm, 
minimum height=1cm, 
text centered, 
text width=9cm, 
draw=black, 
fill = blue!10]

\tikzstyle{decision} = [rectangle, 
minimum width=3cm, 
minimum height=1cm, 
text centered, 
draw=black, 
fill = blue!10]
\tikzstyle{arrow} = [thick,->,>=stealth]
\tikzset{
	block/.style={rectangle, draw, thick, text width=2.5cm, minimum height=1cm, text centered},
	line/.style={draw, -latex'}
}

\theoremstyle{remark}
\newtheorem{rem}{Remark}[section]

\numberwithin{equation}{section}

\usepackage[mathscr]{euscript}
\usepackage{geometry} % to change the page dimensions
\usepackage{graphicx,lscape} % support the \includegraphics command and options
\usepackage{float}
\usepackage{booktabs} % for much better looking tables
\usepackage{bigstrut}
\usepackage{array} % for better arrays (eg matrices) in maths
\usepackage{paralist} % very flexible & customisable lists (eg. enumerate/itemize, etc.)
\usepackage{verbatim} % adds environment for commenting out blocks of text & for better verbatim
\usepackage{subfig} % make it possible to include more than one captioned figure/table in a single float
\biboptions{comma,round,authoryear}

\usepackage{relsize}%
\usepackage{listings}
\usepackage[normalem]{ulem}
\useunder{\uline}{\ul}{}

\usepackage{natbib}

\usepackage{hyperref}
\hypersetup{
colorlinks=true,
allcolors=blue,
pdftitle={A novel INAR(1) process based on the new McKenzie - error thinning operator with geometric marginals},
pdfpagemode=FullScreen,
}

\usepackage{stackengine}

\usepackage[titletoc]{appendix}

\usepackage{enumitem}
\newlist{steps}{enumerate}{1}
\setlist[steps, 1]{label = Step \arabic*:}

\newlist{notes}{enumerate}{1}
\setlist[notes]{label=Note: ,leftmargin=*}

\begin{document}

\begin{frontmatter}
	
	%% Title, authors and addresses
	
	%% use the tnoteref command within \title for footnotes;
	%% use the tnotetext command for the associated footnote;
	%% use the fnref command within \author or \address for footnotes;
	%% use the fntext command for the associated footnote;
	%% use the corref command within \author for corresponding author footnotes;
	%% use the cortext command for the associated footnote;
	%% use the ead command for the email address,
	%% and the form \ead[url] for the home page:
	%%
	%% \title{Title\tnoteref{label1}}
	%% \tnotetext[label1]{}
	%% \author{Name\corref{cor1}\fnref{label2}}
	%% \ead{email address}
	%% \ead[url]{home page}
	%% \fntext[label2]{}
	%% \cortext[cor1]{}
	%% \address{Address\fnref{label3}}
	%% \fntext[label3]{}
	
	\title{\textbf{ Softplus and Neural Architectures for Enhanced Negative Binomial INGARCH Modeling}}
	
	\author[label1]{Divya Kuttenchalil Andrews\corref{cor1}}
	%\ead{divyaandrews5@gmail.com}
	\author[label2]{N. Balakrishna}
	%\cortext[cor1]{Corresponding author}
	\address[label1]{Cochin University of Science and Technology, Kochi, India.}
	\address[label2]{Indian Institute of Technology, Tirupati, India.}

	\author{}
	\address{}
	
	\begin{abstract}
		%% Text of abstract
		The study addresses a significant gap in the literature by introducing the Softplus negative binomial Integer-valued Generalized Autoregressive Conditional Heteroskedasticity (sp NB- INGARCH) model and establishing its stationarity properties, alongside methodology for parameter estimation. Building upon this foundation, the Neural negative binomial INGARCH (neu - NB-INGARCH) model is proposed, designed to enhance predictive accuracy while accommodating moderate non-stationarity in count time series data. A simulation study and data analysis demonstrate the efficacy of the sp NB-INGARCH model, while the practical utility of the neu -	 
		NB - INGARCH model is showcased through a comprehensive analysis of a healthcare data. Additionally, a thorough literature review is presented, focusing on the application of neural networks in time series modeling, with particular emphasis on count time series. In short, this work contributes to advancing the theoretical understanding and practical application of neural network-based models in count time series forecasting. \\
		
		
	\end{abstract}

\end{frontmatter}

%......................................................
\section{Introduction}\label{intro}

Count time series data are pervasive in numerous fields, including public health, economics, and environmental sciences. One of the key challenges in modeling such data is accommodating overdispersion, a phenomenon where the variance exceeds the mean, commonly observed in real-world count data. Among the many models developed to handle overdispersion, the negative binomial distribution stands out due to its flexibility and ability to capture a wide range of dispersion patterns. This flexibility has made it a cornerstone in the modeling of count data, particularly when the Poisson distribution’s constant mean-variance assumption proves inadequate.

In time series contexts, the Integer-valued Generalized Autoregressive Conditional Heteroskedasticity (INGARCH) model is widely used due to its ability to capture temporal dependence in count data \citep{ferland2006integer}. However, linear INGARCH models suffer from certain limitations, notably their inability to represent negative autocorrelation in the data \citep{weiss2022soft}. This limitation often leads to misrepresentations of the underlying dynamics, especially when negative dependencies are present.

To address these shortcomings, \cite{weiss2022softplus} introduced the softplus INGARCH model, leveraging the softplus transformation to ensure positivity of the conditional mean while providing greater flexibility in capturing temporal patterns. Although \cite{weiss2022softplus} primarily explored the properties of the softplus INGARCH model under the Poisson distribution assumption, they extended the model to the negative binomial distribution for data analysis. Notably, the stationarity properties of the negative binomial softplus INGARCH model were considered as a future work, leaving a significant gap in the theoretical development of the model.

Another crucial challenge in count time series modeling arises with data exhibiting non-stationary autocorrelation behaviour, a scenario frequently encountered in daily life, such as hospital admissions, customer arrivals, or online activity. To accommodate non-stationarity, \cite{jahn2023artificial} proposed an extension to the softplus INGARCH framework using artificial neural networks (ANNs), introducing the concept of neural INGARCH models. This approach leverages the flexibility of ANNs to capture complex nonlinear dependencies that traditional models may struggle with. \par 
The integration of ANNs into time series modeling has been explored in various contexts, specially by employing hybrid models. A common technique for such models involve decomposition of a time series into its linear and nonlinear form and applying ANNs to the latter form for better predictive performance. Using this, \cite{time_zhang_2003} proposed a two-step hybrid methodology that combines  Autoregressive Integrated Moving Average (ARIMA) and ANN models and showed promising forecasting efficiency when tested on three data sets.  \cite{medeiros2006building} introduced another hybrid model, integrating an autoregressive (AR) model with a single-hidden-layer ANN. This approach allows for efficient model specification with minimal computational expense. An alternative innovative approach was developed by \cite{hassan2007fusion}, combining a Hidden Markov Model (HMM), Artificial Neural Networks (ANN), and Genetic Algorithms (GA) to forecast financial market behavior.\par \cite{artificial_khashei_2010} revived the ARIMA - ANN combination wherein the first stage of the hybrid methodology involves fitting an ARIMA model; and in the second stage, the residuals from the ARIMA model, along with the original data, are used as inputs to the neural network. \cite{wu2010hybrid} presents a hybrid model that combines an adaptive Wavelet Neural Network (AWNN) with time series models like  Autoregressive Moving Average with Exogenous Inputs (ARMAX) and GARCH to predict daily electricity market values. Their approach produces more accurate forecasts than those reported in previous studies. \cite{gheyas2011novel} incorporates and combines a regression neural network model with multiple machine learning algorithms. This hybrid model leverages the combined strengths of the algorithms, though it comes with a high computational cost. \par \cite{young2015predicting} further advanced the framework to build a three - dimensional (3D) hydrodynamic model blended with an ANN model (using backpropagation neural network (BPNN)) and an ARMAX model to more accurately predict water level fluctuations. The innovative approach of combining the 3D hydrodynamic model with an ANN has demonstrated enhanced prediction accuracy for water level changes. \cite{wang2016financial} proposed the combination of Elman Recurrent Neural Networks (ERNN) with a Stochastic Time Effective Function (STNN), i.e., the ERNN-STNN model. The empirical results demonstrate that this neural network outperforms linear regression, Complexity Invariant Distance (CID), and Multi-Scale CID (MCID) analysis methods. When compared to other models, such as BPNN, it shows superior performance in financial time series forecasting.  \cite{rahayu2017hybrid} focussed on  a hybrid of time series regression, ARIMA and ANN to provide better forecasts of currency flow data at Bank Indonesia. \cite{ahmedteal} conducted an extensive systematic review and comparison of neural network models for time series forecasting, which were published in the period 2006 - 2016. \par
Eventhough a lot many works have come up with regard to application of neural networks in continuous time series using normal assumption, not much have come up in the area of integer-valued time series. Two interesting works that paved path to such research were the softplus response for unbounded counts \citep{weiss2022softplus} and the soft-clipping response for bounded counts \citep{weiss2022soft}, which produce nearly linear models while still allowing for negative autocorrelation.  This subsequently led to the work by \cite{jahn2023artificial} on nonlinear INGARCH models where the response function is represented by a single hidden layer feedforward artificial neural network, referred to as a neural INGARCH model with conditional distributions as Poisson for unbounded count time series. Further, \cite{flexible_jahn_2024}  introduced a novel framework that combines spatio-temporal regression techniques with Artificial Neural Networks (ANNs) by leveraging ANN's universal approximation property to model arbitrary spatial patterns using geographic coordinates as regressors. 
Despite these advancements, most studies utilizing neural network-based count models have majorly focused on the Poisson distribution. Applications involving the negative binomial distribution remain relatively scarce, despite its importance in modeling overdispersed count data. Moreover, the theoretical underpinnings of neural network-based negative binomial INGARCH models remain underdeveloped, limiting their broader adoption in practice.

In this paper, we extend the scope of the existing literature by developing the theoretical properties of the softplus negative binomial INGARCH model, including its stationarity. We define and analyze the neural negative binomial INGARCH model, providing theoretical insights and applying it to real-world non-stationary count time series. To showcase the practical relevance of these models, we conduct two data analyses - a study using weekly Syphilis counts and the second using hourly emergency department arrivals at a hospital, a highly relevant application where overdispersion and non-stationarity are prominent.\par

By addressing these gaps, we aim to underscore the utility of the negative binomial distribution in capturing the complex dynamics of count time series and highlight the advantages of integrating ANN-based approaches in this domain.

\section{A Softplus INGARCH model}\label{spn}

While INGARCH models exhibit an ARMA-like autocorrelation structure, the range of autocorrelation function (ACF) of the model is often more restricted than that of standard ARMA models, as negative ACF values are generally not feasible. This limitation arises due to parameter constraints necessary to ensure nonnegative values (counts) in the process. To include the possibility of negative ACF values, conditional regression models with a log link might be considered; however, this approach sacrifices both the linear conditional mean and the ARMA-like ACF structure. Additionally, the lack of analytic expressions for the mean, variance, and ACF in a log-linear INGARCH model complicates the estimation procedure.

To address this issue, \cite{weiss2022softplus} introduced a new family of conditional regression models for stationary count processes \( \{X_t\}_{t \in \mathbb{Z}} \) that applies a softplus function to link the conditional mean \( \lambda_t = E(X_t \mid X_{t-1}, \dots) \) to a linear expression involving past observations \( X_{t-i} \) and previous conditional means \( \lambda_{t-j} \). In particular, they considered the softplus Poisson INGARCH model:
\begin{equation}
	\label{sppois}
	X_t \mid \mathscr{F}_{t-1} \sim  Poi(\lambda_t),  \quad
	\lambda_t=sp\left(\alpha_0+\sum_{i=1}^p\alpha_i X_{t-i}+\sum_{j=1}^{q}\beta_j \lambda_{t-j}\right), \; \forall t \in \mathbb{Z},
\end{equation}
where sp() denotes the softplus function (to be discussed at length in the following subsection), 
$\mathbb{Z} = \{ \ldots -2,-1,0,1,2,\ldots\}$  is the set of integers,  $\mathscr{F}_{t-1}$ is the $\sigma$-field generated by $\{X_{t-1}, X_{t-2},\ldots\}$, $ \alpha_i, \beta_j$ are real-valued for $i= 0,1,\ldots ,p, \; j=1, \ldots q$,  and $Poi$ represents the Poisson distribution.  \cite{zhusp} studied the properties of a softplus model assuming beta negative binomial (BNB) as the conditional distribution focussing on the distribution's heavy tail properties. However, no study concerning the simple negative binomial distribution has come up, despite the distribution gaining large prominence in the context of count time series models in the past decade. So, the present article is aimed at deriving properties of the softplus negative binomial distribution and extending it to an artificial neural network INGARCH  setup. We now look into a brief literature surrounding the softplus function and some of its useful properties.
\subsection{The softplus function}
\cite{dugas2000incorporating} proposed the softplus function as $s(x) = ln(1+exp(x)), \; \forall x \in \mathbb{R}$, where $\mathbb{R}$ is the set of real numbers. \cite{zhang2018permuted} noted that the function is positive, continuous and differentiable on  $\mathbb{R}$, indicating its potential in multinomial regression. The function is also frequently viewed as a smoother alternative to the rectified linear unit function, defined as \( \text{ReLU}(x) = \max\{0, x\} \), considered to be one of the primary nonlinear activation function for deep neural networks \citep{nair2010rectified}. The softplus function used in the present article is the generalized version of the function (See \cite{mei2017neural}):
\begin{equation}
	\label{spfun}
	sp(x) = c\; ln\left(1+exp\left(\frac{x}{c}\right)\right),
\end{equation}
the extra tuning parameter, $c>0$, was introduced to control the deviation between $sp$ and $ReLU$ functions and by default, we assume $c=1$. Interestingly, when $c=1$, the first derivative of $sp$ with respect to $x$, $sp'(x)= \frac{1}{1+exp(-x)}$, is the logistic function (or the sigmoid function) which is also a popular activation function in neural networks. \cite{weiss2022softplus} stated that $sp(x)$ has additional desirable properties:
\begin{enumerate}
	\item[(i)] As the value of $x$ increases, the deviation between $sp(x)$ and $ReLU(x)$ decreases. In particular, $\lim_{c \to 0} sp(x) = ReLU(x)$.
	\item[(ii)] $sp(x)$ is Lipschitz with Lipschitz constant $1$.
	\item[(iii)] The softplus function $sp(x)$ is strictly monotone increasing with $sp(0) = c\; ln(2)$ and $sp'(x) < 1$.  Thus, it can be bounded from above by $ c\; ln(2) + max\{0,x\}$. In fact, the following inequality holds:
	\begin{equation}
		\label{ineq}
		ReLU(x) < sp(x) \leq c\;ln(2) + ReLU(x). 
	\end{equation}
\end{enumerate}
These properties are useful for proving the first and second-order stationarity conditions for the softplus INGARCH model. Next, we define the softplus negative binomial INGARCH model.
\subsection{A softplus negative binomial INGARCH model}
\cite{zhu2011negative} introduced the linear negative binomial INGARCH model with the conditional distribution $NB(n,p_t)$ defined as:
\begin{equation}
	\label{nb}
	Pr[X_t=x \mid \mathscr{F}_{t-1}] = \binom{x + n - 1}{n - 1} p_t^n (1 - p_t)^x,\; x= 0, 1, 2, \ldots,
\end{equation}
where $n>0$  is fixed over time and $0<p_t \leq1$.

\begin{defn}
	If the conditional distribution of $\{X_t\}$ is defined as in (\ref{nb}), i.e., $X_t\mid \mathscr{F}_{t-1} \sim NB(n,p_t)$, and
	\begin{equation}
		\label{spnbrec}
		\lambda_t= E[X_t\mid \mathscr{F}_{t-1}]=sp\left(\alpha_0+\sum_{i=1}^p\alpha_i X_{t-i}+\sum_{j=1}^{q}\beta_j \lambda_{t-j}\right), \; \forall t \in \mathbb{Z},
	\end{equation}
	and the coefficients reprise their definitions specified in (\ref{sppois}) and (\ref{spfun}), then the model is called \textit{softplus negative binomial INGARCH} model and denoted by \textit{sp NB-INGARCH(p,q)}. 
\end{defn}
\noindent Naturally, we obtain the conditional variance of $X_t$ as 
\begin{equation}
	\label{cdvar}
	Var[X_t\mid \mathscr{F}_{t-1}] = \frac{n(1-p_t)}{p_t^2} = \lambda_t\left(1+ \frac{\lambda_t}{n}\right).
\end{equation}
With this information, by properties of expectation, we can proceed to compute the unconditional mean and variance respectively as: 
\begin{equation}
	\label{exp}
	\begin{aligned}
		E[X_t] &= E[\lambda_t],\\ 
		Var[X_t] &= E[Var[X_t \mid \mathscr{F}_{t-1}]] + Var[E[X_t \mid \mathscr{F}_{t-1}]]\\ 
		&= E\left[\lambda_t\left(1+\frac{\lambda_t}{n}\right)\right] + Var[\lambda_t]\\ 
		& \quad > E[\lambda_t].
	\end{aligned}
\end{equation}
Hence, the model can very well represent overdispersed count time series. Note that as $n\to \infty$, $p_t \to 1$, and $n(1-p_t) \to \mu$, the negative binomial distribution converges to a Poisson distribution with mean $\mu$ \citep{casellab}. Through the following theorems, the stationarity and existence of first and second order moments for the softplus NB - INGARCH model are stated and proved. For the sake of simplicity, we will limit our focus to the first-order case where $p=q=1$.
\begin{thm}
	\label{thm1}
	Consider the sp NB-INGARCH(1,1) model. If $\overline{\alpha_1} + \overline{\beta_1} < 1$ and $|\beta_1| < 1 $, then there is a stationary ergodic process  $\{X_t\}_{t \in \mathbb{Z}}$ that satisfies (\ref{spnbrec}) and (\ref{exp}) and has $E[X_t] < \infty$, where $\overline{\alpha_1} = max\{0,\alpha_1\}$, $\overline{\beta_1} = max\{0,\beta_1\}$.
\end{thm}
\begin{proof}
	We prove the existence of a stationary ergodic solution by verifying the three assumptions $\mathcal{A}1$ to $\mathcal{A}3$ in \cite{doukhan2019absolute} (See \textcolor{blue}{Appendix} \ref{App2}). It can be seen from (\ref{ineq}) that 
	\begin{equation}
		sp(x) \leq c\;ln(2) + max\{0,x\} \leq c\;ln(2) + \overline{\alpha_0} + \overline{\alpha_1} X_{t-1} + \overline{\beta_1}\lambda_{t-1}, \nonumber
	\end{equation}
	where $\overline{\alpha_0} = max\{0,\alpha_0\}$. Let $b_0 =1$, $a_0=c\;ln(2) + \overline{\alpha_0}$ and $\overline{\alpha_1} + \overline{\beta_1} \leq k <1$. Then, for $Y_t = \{X_{t-1},\ldots,X_{t-p+1},\lambda_t,\ldots,\lambda_{t-q+1}\}$ which reduces to $Y_t = \lambda_t$ for $p=q=1$, we have
	\begin{equation}
		E[Var[Y_t\mid Y_{t-1}]] \leq b_0(c\;ln(2) + \overline{\alpha_0} + (\overline{\alpha_1}+\overline{\beta_1})\lambda_{t-1}) \leq a_0 + k Var[Y_{t-1}], \nonumber
	\end{equation}
	where $Var[Y_t] = b_0\lambda_t$, so the condition $\mathcal{A}1$ is satisfied. Since $sp(x)$ satisfies Lipschitz condition, for any $x_1 \in \mathbb{R}$, $l_1, l_1' \geq 0$, we can get
	\begin{equation}
		\left| sp(\alpha_0 + \alpha_1x_1 + \beta_1l_1)- sp(\alpha_0 + \alpha_1x_1 + \beta_1l_1')\right| \leq |\beta_1||l_1 - l_1'|, \nonumber
	\end{equation}
	and hence condition $\mathcal{A}2$ is satisfied. Finally, by a part of the proof of Lemma 7 in  \cite{gorgi2020beta} (See \textcolor{blue}{Appendix} \ref{App3}),
	we have
	\begin{equation}
		\label{a3org}
		TV[NB_{p_1}, NB_{p_2}] \leq 1 - exp\left(-n\left| \frac{1-p_1}{p_1} - \frac{1-p_2}{p_2}\right|\right), 
	\end{equation}
	where  $NB_p$ denotes the probability measure of a negative binomial distribution with parameters $n$ and $p$, $TV[.]$ stands for total variation distance between $NB_{p_1}$ and $NB_{p_2}$ and so condition $\mathcal{A}3$ is established. Thus, stationary ergodic solutions exist. Further, since $sp(x) \leq c\; ln(2) + \overline{\alpha_0}+\overline{\alpha_1}X_{t-1}+\overline{\beta_1}\lambda_{t-1}$, we can get
	\begin{equation}
		\begin{aligned}
			E_{t-2}[\lambda_t] &\leq c\; ln(2) + \overline{\alpha_0} + \overline{\alpha_1}E_{t-2}[X_{t-1}]+ \overline{\beta_1}\lambda_{t-1}\\
			& \quad =\bar{\mathcal{C}}_{10} + \bar{\mathcal{C}}_{11}\lambda_{t-1}, \nonumber
		\end{aligned}
	\end{equation}
	where $E_t[.] = E[.| \mathscr{F}_t]$, $\bar{\mathcal{C}}_{10} = c\;ln(2) + \overline{\alpha_0}$ and $\bar{\mathcal{C}}_{11} = \overline{\alpha_1}+\overline{\beta_1}$.
	From the conditional expectation property $E[\lambda_t\mid \mathscr{F}_{t-3}] = E[E[\lambda_t\mid\mathscr{F}_{t-2}]\mid\mathscr{F}_{t-3}]$, by recursion we obtain
	\begin{equation}
		\label{cdrec}
		E_{t-h}[\lambda_t] \leq \sum_{i=0}^{h-2}\bar{\mathcal{C}}_{11}^{i}\bar{\mathcal{C}}_{10} + \bar{\mathcal{C}}_{11}^{h-1}\lambda_{t-h+1}. 
	\end{equation}
	Since $\bar{\mathcal{C}}_{11} < 1$, as $h \to \infty$, $E[\lambda_t] \leq (1-\bar{\mathcal{C}}_{11})\bar{\mathcal{C}}_{10}$, which implies $E[\lambda_t] < \infty$. Thus, by (\ref{exp}), $E[X_t] < \infty$.
	Hence, the proof. 
\end{proof}
\begin{thm}
	\label{thm2}
	Consider the softplus NB-INGARCH(1,1) model. If 
	\begin{equation}
		\overline{\alpha_1} + \overline{\beta_1} <1, \; |\beta_1| <1,\; \text{and} \; \left(1+\frac{1}{n}\right) \overline{\alpha_1}^2 + 2\overline{\alpha_1}\overline{\beta_1} + \overline{\beta_1}^2 <1, \nonumber
	\end{equation}
	then $E[X_t^2] < \infty$.
\end{thm}
\begin{proof}
	As a consequence of $\lambda_t \leq c\;ln(2) + \overline{\alpha_0} + \overline{\alpha_1}X_{t-1} + \overline{\beta_1}\lambda_{t-1}$, we have
	\begin{equation}
		\begin{aligned}
			\lambda_t^2  &\leq (c\;ln(2) + \overline{\alpha_0})^2 + \overline{\alpha_1}^2X_{t-1}^2 + \overline{\beta_1}^2\lambda_{t-1}^2 + 2 \overline{\alpha_1}\overline{\beta_1}X_{t-1}\lambda_{t-1}\\
			&\quad + 2(c\;ln(2) +\overline{\alpha_0})(\overline{\alpha_1}X_{t-1}+ \overline{\beta_1}\lambda_{t-1}). \nonumber
		\end{aligned}
	\end{equation}
	From (\ref{cdvar}), we obtain $	E_{t-1}[X_t^2]  = \lambda_t + (1+1/n) \lambda_t^2$. Therefore,
	\begin{equation}
		\begin{aligned}
			E_{t-2}[\lambda_t^2]  &\leq  (c\;ln(2) + \overline{\alpha_0})^2 + (\overline{\alpha_1}^2 + 2(c\;ln(2) + \overline{\alpha_0})(\overline{\alpha_1}+\overline{\beta_1}))\lambda_{t-1}\\
			&\quad + \left(\overline{\alpha_1}^2\left(1+\frac{1}{n}\right) + 2\overline{\alpha_1}\overline{\beta_1}+\overline{\beta_1}^2\right)\lambda_{t-1}^2 \nonumber\\
			& \quad =\bar{\mathcal{C}}_{20} + \bar{\mathcal{C}}_{21}\lambda_{t-1}+\bar{\mathcal{C}}_{22}\lambda_{t-1}^2, 
		\end{aligned}
	\end{equation}
	where $\bar{\mathcal{C}}_{20} = (c\;ln(2)+\overline{\alpha_0})^2$, $\bar{\mathcal{C}}_{21} = \overline{\alpha_1}^2 + 2(\overline{\alpha_1}+\overline{\beta_1})(c\;ln(2)+\overline{\alpha_0})$ and $\mathcal{C}_{22} = \left(1+\frac{1}{n}\right)\overline{\alpha_1}^2 + 2\overline{\alpha_1}\overline{\beta_1}+\overline{\beta_1}^2$. Hence, we can write
	\begin{equation}
		E_{t-2}[M_t] \leq \mathcal{D}_0 + \mathcal{D}_1M_{t-1}, \nonumber
	\end{equation}
	where $M_t = \left( \begin{array}{c}
		\lambda_t \\
		\lambda_t^2
	\end{array} \right)
	$, $\mathcal{D}_0 = \left(\begin{array}{c}
		\bar{\mathcal{C}}_{10} \\
		\bar{\mathcal{C}}_{20} 
	\end{array} \right)
	$, and $\mathcal{D}_1 = \left(\begin{array}{c c}
		\bar{\mathcal{C}}_{11}  & 0\\
		\bar{\mathcal{C}}_{21} & \bar{\mathcal{C}}_{22}
	\end{array} \right)
	$. By (\ref{cdrec}), we get
	\begin{equation}
		E_{t-h}[M_t] \leq \sum_{i=0}^{h-2} \mathcal{D}_1^i\mathcal{D}_1^{h-1}M_{t-h+1}, \nonumber
	\end{equation}
	where `$\leq$' is satisfied  by the corresponding components of the vectors on the left and right hand side of the inequality. If the assumptions stated in the theorem are valid, it directly implies that the diagonal elements of the lower triangular matrix $\mathcal{D}_1$ are $\mathcal{C}_{11} <1$, and $\mathcal{C}_{22} < 1$. As $h \to \infty$, we get
	\begin{equation}
		E[M_t] \leq (I - \mathcal{D}_1)^{-1}\mathcal{D}_0, \nonumber
	\end{equation}
	where $I$ is the $2 \times 2$ unit matrix. Thus, $E[\lambda_t^2] < \infty$, and so $E[X_t^2] < \infty$. The proof is now complete.
\end{proof}
\subsection{Estimating moments through Approximation}\label{approx}
While \textcolor{blue}{Theorems} \ref{thm1} and \ref{thm2} establish the existence of moments for the softplus INGARCH model, deriving exact closed-form expressions for the moments of  the model presents significant challenges. However, because the softplus function closely resembles the piecewise linear ReLU function, the moment properties of the softplus INGARCH model can be approximated using the moment formulae of the linear INGARCH model \citep{weiss2022softplus}. Leveraging the moment formulae from NB-INGARCH model as approximations allows us to gain insights into the mean, variance, and autocorrelation function (ACF) of the sp NB-INGARCH model. \par
The linear NB-INGARCH(p,q) model is defined as:
\begin{equation}
	\label{nbing}
	X_t \mid \mathscr{F}_{t-1} \sim  NB(n,p_t),  \quad
	\lambda_t=\frac{n(1-p_t)}{p_t} =\mathfrak{a}_0+\sum_{i=1}^p \mathfrak{a}_i X_{t-i}+\sum_{j=1}^{q} \mathfrak{b}_j \lambda_{t-j}, 
\end{equation}
where $\mathfrak{a}_0 > 0$, $\mathfrak{a}_i \geq 0$, $\mathfrak{b}_j \geq 0$, for $i=1,\ldots p$, $j = 1, \ldots q$, $p \geq 1$, $q \geq 1$. $\mathscr{F}_{t-1}$ reprises the definition from (\ref{sppois}). The model (\ref{nbing}) slightly differs from the one proposed by \cite{zhu2011negative} in that the latter considers $\lambda_t = \frac{1-p_t}{p_t}$.  Incorporating this minor alteration to the theorems of stationarity and recurrence relation satisfied by autocovariance of $\{X_t\}$ (denoted as $\gamma_X(.)$) and $\{\lambda_t\}$ (denoted as $\gamma_\lambda(.)$) which are proved in \cite{zhu2011negative}, we can easily obtain the following results.

\begin{res}
	\label{thm3}
	If the stochastic process \( \{X_t\}_{t \in \mathbb{Z}} \) of the NB-INGARCH(\(p, q\)) model is stationary, then its mean is:
	\[
	\mu = \frac{\mathfrak{a}_0}{1 - \sum_{i=1}^{p} \mathfrak{a}_i - \sum_{j=1}^{q} \mathfrak{b}_j}.
	\]
\end{res}
\begin{res}
	\label{thm4}
	Suppose that the process $\{X_t\}$ following NB-INARCH(p) model is first - order stationary. Then, a necessary and sufficient condition for the process to be second - order stationary is that 
	\begin{equation}
		1- A_1b^{-1} - \ldots - A_p b^{-p} = 0 \nonumber
	\end{equation}
	has all roots lying inside the unit circle, where for $r,s= 1, \ldots, p-1$, assuming $\omega_0= 1+\frac{1}{n}$,
	$A_r =\omega_0\left(\mathfrak{a}_r^2 - \sum_{v=1}^{p-1}\sum_{|i-j|=v} \mathfrak{a}_i\mathfrak{a}_jb_{vr}\nu_{r0}\right)$, $A_p =  \omega_0\mathfrak{a}_p^2$, $\nu_{s0} = \mathfrak{a}_s$, $\nu_{ss} = \sum_{|i-s|=s} \mathfrak{a}_i -1$ and $\nu_{sr} = \sum_{|i-s|=r}\mathfrak{a}_i, r \neq s$. 
\end{res} 
\begin{res}
	\label{thm5}
	Suppose that $\{X_t\}$ following NB-INGARCH(p,q) process is second - order stationary. Then, the autocovariance functions corresponding to $\{X_t\}$ and $\{\lambda_t\}$ respectively satisfy the equations
	\smaller
	\begin{equation}
		\begin{aligned}
			\gamma_X(h) &= \sum_{i=1}^p\mathfrak{a}_i\gamma_X(|h-i|) + \sum_{j=1}^{min(h-1,q)}\mathfrak{b}_j\gamma_X(h-j) +  \sum_{j=h}^q\mathfrak{b}_j\gamma_{\lambda}(j-h), \quad h \geq 1; \nonumber \\
			\gamma_{\lambda}(h) &=  \sum_{i=1}^{min(h,p)}\mathfrak{a}_i\gamma_{\lambda}(|h-i|) +  \sum_{i=h+1}^p\mathfrak{a}_i\gamma_X(i-h) + \sum_{j=1}^q \mathfrak{b}_j \gamma_{\lambda}(|h-j|), \quad h \geq 0.
		\end{aligned}
	\end{equation}
\end{res} 
The proofs of the above results can be arrived at by applying the reparametrization in \cite{zhu2011negative}. Alternatively, one can also refer to the proofs in  \cite{bayesian}. 
\begin{rem}
	\label{remo}
	From \textcolor{blue}{Results} \ref{thm3} to \ref{thm5}, we can readily obtain the following equations:
	
	\begin{align}
		\label{eqs}
		\gamma_X(0) &= \omega_0\gamma_\lambda(0) +  \mu\left(1+ \frac{\mu}{n}\right) , \quad \gamma_\lambda(0) = \sum_{i=1}^p \mathfrak{a}_i \gamma_\lambda(i) + \sum_{j=1}^q \mathfrak{b}_j \gamma_\lambda(j),  \\
		\label{eqs1}
		\gamma_X(h) &= \sum_{i=1}^{\min(h,p)} \mathfrak{a}_i \gamma_X(h - i) + \sum_{j=h}^q \mathfrak{b}_j \gamma_X(h - j) + \gamma_\lambda(h), \quad h \geq 1, \\
		\label{eqs2}
		\gamma_\lambda(h) &= \sum_{i=1}^p \mathfrak{a}_i \gamma_\lambda(h - i) + \sum_{j=h}^q \mathfrak{b}_j \gamma_\lambda(h - j), \quad h \geq 1,
	\end{align}
	
	where $\omega_0= 1+\frac{1}{n}$,  $\gamma_X(h) = \text{Cov}(X_t, X_{t-h})$, and $\gamma_\lambda(h) = \text{Cov}(\lambda_t, \lambda_{t-h})$.
\end{rem}
In the following example, the statistical properties of NB-INGARCH (1,1) model are stated based on the \textcolor{blue}{Results} \ref{thm3}, \ref{thm4}, \ref{thm5} and \textcolor{blue}{Remark} \ref{remo}.
\begin{ex}
	\label{exple}
	The mean ($\mu$), variance ($\gamma_X(0)$), and ACF ($\rho_X(h)$) of NB-INGARCH (1,1) respectively are:
	\begin{align}
		\mu = \frac{\mathfrak{a}_0}{1 - \mathfrak{a}_1 - \mathfrak{b}_1}, 
	\end{align}
	\begin{align}
		\gamma_X(0)= \mu\left(1+\frac{\mu}{n}\right)\left(\frac{1-2\mathfrak{a}_1\mathfrak{b}_1 - \mathfrak{b}_1^2}{1-\omega_0\mathfrak{a}_1^2 - 2\mathfrak{a}_1\mathfrak{b}_1 - \mathfrak{b}_1^2}\right),
	\end{align}
	and
	\begin{align}
		\rho_X(h) = \mathfrak{a}_1(\mathfrak{a}_1+\mathfrak{b}_1)^{h-1}\left(\frac{1-\mathfrak{a}_1\mathfrak{b}_1 - \mathfrak{b}_1^2}{1-2\mathfrak{a}_1\mathfrak{b}_1 - \mathfrak{b}_1^2}\right),\quad h \geq 1. 
	\end{align}
\end{ex}
To approximate the moments of sp NB - INGARCH, we can substitute the parameters of linear NB-INGARCH models with the parameters of former, i.e., substitute $\mathfrak{a}_0$,  $\mathfrak{a}_1$ and $\mathfrak{b}_1$, by $\alpha_0$, $\alpha_1$ and $\beta_1$ respectively. According to numerical studies conducted by \cite{weiss2022softplus} and \cite{zhusp}, such approximations are very close to the empirical moments of data generated from their respective softplus models for most cases. Nevertheless, to test the accuracy of the approximations for NB-INGARCH, we perform a simulation exercise with different parameter combinations for sp NB- INGARCH(1,1) model. We considered three values of the mean $\mu$, i.e., $\mu = 2, 6,$ and $12$, for selection of parameter combinations and a sample of size $10^6$ was generated for each configuration. We compared the true moment values with the approximate "linear" moment values obtained from the linear NB-INGARCH model given in  \textcolor{blue}{Example} \ref{exple}. The results of our study are summarized in \textcolor{blue}{Tables} \ref{tab1} and \ref{tab2}.\\ 
Note that for \textcolor{blue}{Table} \ref{tab2}, $c\to 0$ indicates the ReLU INGARCH model (See \citet[Remark~1]{weiss2022softplus}).

Observing results for Models $1$ and $13$, we notice that the true and approximate mean are quite different. This is because of a low intercept value $\alpha _0$ in the softplus function. To make the approximation better, it is recommended to use a smaller value of $c$. By using values like $c=0.5, c=0.25$, or even $c\to 0$, we can see an improvement in the linearity of the model. Similar deviations are observed in Models $12$, $18$ and $22$, especially in certain parameters like dispersion ratio and ACFs of lags $2$ and $3$. By reducing the value of c, we can make the approximation better. However, even with very small values of c, there are still some deviations because the functions used are not strictly linear. This anomaly could also be due to strong negative autocorrelation in Models 12 and 22.\par

In conclusion, the sp NB INGARCH model can approximate the mean, variance, and ACF well under certain conditions. By adjusting the value of $c$, we can improve the accuracy of the model. Overall, as long as the stationarity conditions are met and the autocorrelation, if negative, is not too strong, the softplus INGARCH model can provide a good approximation of the true values. The approximations also have the advantage of better representation of negative ACF values as compared to the linear NB-INGARCH model. Hence, the approximations can be very well used as initial values for conditional maximum likelihood estimation.

\subsubsection{Conditional maximum likelihood estimation}\label{est}

To estimate the parameters of the softplus NB-INGARCH model by the maximum likelihood method, we need to estimate the parameter vector $\boldsymbol{\Theta} = (\theta, n)^\top$, where the parameter space is $\boldsymbol{\Theta}$, and the initial parameter value is $\Theta_0 = (\theta_0, n_0)$. Here, $\theta = (\alpha_0, \alpha_1,\ldots \alpha_p, \beta_1, \ldots \beta_p)$ represents the parameter space, and $\{x_t\}_{t=1}^s$ denotes the observed values of $X_t$ from a sample of size $s$.

Based on the model definition (\ref{nb}) and (\ref{spnbrec}), we can easily obtain the log-likelihood function:
\begin{equation}
	\label{lik}
	\begin{aligned}
		l(\boldsymbol{\Theta}) &= n\sum_{t=1}^{s}log p_t + \sum_{t=1}^sx_tlog(1-p_t) + log \prod_{t=1}^{s} \binom{x_t + n - 1}{n - 1} \\
		&= \sum_{t=1}^s\Big\{x_t log\left(\frac{\lambda_t}{n}\right) - (n+x_t)\log\left(1+\frac{\lambda_t}{n}\right)+ \sum_{v=1}^{x_t}log(v+n-1) - log(x_t!)\Big\}.
	\end{aligned}
\end{equation}
An intial value of $\lambda_t$ is obtained by taking the average value of $\{x_t\}_{t=1}^s$ and consequently substituting in the softplus function. Maximizing the log-likelihood function (\ref{lik}) through numerical optimisation methods, we obtain the conditional maximum likelihood estimator (CMLE) $\boldsymbol{\hat{\Theta}}$ of $\boldsymbol{\Theta}$.

\begin{table}[H]
	\center
	\caption{Comparison of moments from  sp INGARCH(1,1) models with corresponding moments obtained using linear model formulae(``lin"), for $c=1$ and $c=0.5$.}
	\renewcommand{\arraystretch}{1.2}
	\scalebox{0.55}{%
		\begin{tabular}{cccccccccccccc}
			\multicolumn{1}{l|}{}      & \multicolumn{1}{l}{}           & \multicolumn{1}{l}{}           & \multicolumn{1}{l|}{}          & \multicolumn{2}{c|}{$\mu$}                           & \multicolumn{2}{c|}{$\sigma ^2 / \mu$} & \multicolumn{2}{c|}{$\rho (1)$}      & \multicolumn{2}{c|}{$\rho (2)$}      & \multicolumn{2}{c}{$\rho (3)$} \\
			\multicolumn{1}{l|}{Model} & \multicolumn{1}{c}{$\alpha_0$} & \multicolumn{1}{c}{$\alpha_1$} & \multicolumn{1}{c|}{$\beta_1$} & \multicolumn{1}{l}{sp} & \multicolumn{1}{c|}{lin}    & sp      & \multicolumn{1}{c|}{lin}     & sp     & \multicolumn{1}{c|}{lin}    & sp     & \multicolumn{1}{c|}{lin}    & sp             & lin           \\ \hline
			\multicolumn{14}{c}{$c=1$, $n=3$}                                                                                                                                                                                                                                                                                                            \\ \hline
			\multicolumn{1}{c|}{1}     & 0.6                            & 0.3                            & \multicolumn{1}{c|}{0.4}       & 2.403                  & \multicolumn{1}{c|}{2.002}  & 2.140   & \multicolumn{1}{c|}{2.079}   & 0.323  & \multicolumn{1}{c|}{0.360}  & 0.210  & \multicolumn{1}{c|}{0.252}  & 0.137          & 0.176         \\
			\multicolumn{1}{c|}{2}     & 1.8                            & 0.3                            & \multicolumn{1}{c|}{0.4}       & 6.008                  & \multicolumn{1}{c|}{6.000}  & 3.739   & \multicolumn{1}{c|}{3.750}   & 0.358  & \multicolumn{1}{c|}{0.360}  & 0.250  & \multicolumn{1}{c|}{0.252}  & 0.173          & 0.176         \\
			\multicolumn{1}{c|}{3}     & 3.6                            & 0.3                            & \multicolumn{1}{c|}{0.4}       & 12.010                 & \multicolumn{1}{c|}{12.000} & 2.395   & \multicolumn{1}{c|}{2.388}   & 0.362  & \multicolumn{1}{c|}{0.360}  & 0.253  & \multicolumn{1}{c|}{0.252}  & 0.176          & 0.176         \\ \hline
			\multicolumn{1}{c|}{4}     & 2.2                            & 0.3                            & \multicolumn{1}{c|}{-0.4}      & 2.131                  & \multicolumn{1}{c|}{1.907}  & 1.873   & \multicolumn{1}{c|}{1.847}   & 0.243  & \multicolumn{1}{c|}{0.312}  & -0.020 & \multicolumn{1}{c|}{-0.098} & 0.003          & 0.030         \\
			\multicolumn{1}{c|}{5}     & 6.6                            & 0.3                            & \multicolumn{1}{c|}{-0.4}      & 6.004                  & \multicolumn{1}{c|}{6.132}  & 3.385   & \multicolumn{1}{c|}{2.946}   & 0.266  & \multicolumn{1}{c|}{0.325}  & -0.027 & \multicolumn{1}{c|}{-0.106} & 0.002          & 0.034         \\
			\multicolumn{1}{c|}{6}     & 13.2                           & 0.3                            & \multicolumn{1}{c|}{-0.4}      & 12.001                 & \multicolumn{1}{c|}{11.948} & 2.196   & \multicolumn{1}{c|}{2.363}   & 0.268  & \multicolumn{1}{c|}{0.153}  & -0.025 & \multicolumn{1}{c|}{-0.023} & 0.002          & 0.004         \\ \hline
			\multicolumn{1}{c|}{7}     & 1.8                            & -0.3                           & \multicolumn{1}{c|}{0.4}       & 2.163                  & \multicolumn{1}{c|}{2.145}  & 1.863   & \multicolumn{1}{c|}{1.604}   & -0.226 & \multicolumn{1}{c|}{-0.200} & -0.022 & \multicolumn{1}{c|}{0.020}  & -0.002         & 0.000         \\
			\multicolumn{1}{c|}{8}     & 5.4                            & -0.3                           & \multicolumn{1}{c|}{0.4}       & 6.010                  & \multicolumn{1}{c|}{5.906}  & 3.368   & \multicolumn{1}{c|}{3.562}   & -0.262 & \multicolumn{1}{c|}{-0.200} & -0.027 & \multicolumn{1}{c|}{0.020}  & -0.002         & 0.000         \\
			\multicolumn{1}{c|}{9}     & 10.8                           & -0.3                           & \multicolumn{1}{c|}{0.4}       & 11.998                 & \multicolumn{1}{c|}{12.009} & 2.200   & \multicolumn{1}{c|}{2.195}   & -0.267 & \multicolumn{1}{c|}{-0.200} & -0.024 & \multicolumn{1}{c|}{0.020}  & -0.005         & 0.000         \\ \hline
			\multicolumn{1}{c|}{10}    & 3.4                            & -0.3                           & \multicolumn{1}{c|}{-0.4}      & 2.097                  & \multicolumn{1}{c|}{2.811}  & 1.897   & \multicolumn{1}{c|}{2.250}   & -0.262 & \multicolumn{1}{c|}{-0.300} & 0.158  & \multicolumn{1}{c|}{0.200}  & -0.093         & 0.000         \\
			\multicolumn{1}{c|}{11}    & 10.2                           & -0.3                           & \multicolumn{1}{c|}{-0.4}      & 6.016                  & \multicolumn{1}{c|}{7.174}  & 3.662   & \multicolumn{1}{c|}{4.287}   & -0.340 & \multicolumn{1}{c|}{-0.300} & 0.234  & \multicolumn{1}{c|}{0.200}  & -0.159         & 0.000         \\
			\multicolumn{1}{c|}{12}    & 20.4                           & -0.3                           & \multicolumn{1}{c|}{-0.4}      & 11.997                 & \multicolumn{1}{c|}{11.153} & 2.381   & \multicolumn{1}{c|}{6.598}   & -0.359 & \multicolumn{1}{c|}{-0.300} & 0.251  & \multicolumn{1}{c|}{0.200}  & -0.175         & 0.000         \\ \hline
			\multicolumn{1}{c|}{13}    & 0.6                            & 0.4                            & \multicolumn{1}{c|}{0.3}       & 2.447                  & \multicolumn{1}{c|}{2.002}  & 2.480   & \multicolumn{1}{c|}{2.446}   & 0.433  & \multicolumn{1}{c|}{0.471}  & 0.284  & \multicolumn{1}{c|}{0.330}  & 0.188          & 0.231         \\
			\multicolumn{1}{c|}{14}    & 1.8                            & 0.4                            & \multicolumn{1}{c|}{0.3}       & 6.026                  & \multicolumn{1}{c|}{6.001}  & 4.410   & \multicolumn{1}{c|}{4.403}   & 0.469  & \multicolumn{1}{c|}{0.472}  & 0.328  & \multicolumn{1}{c|}{0.330}  & 0.229          & 0.231         \\
			\multicolumn{1}{c|}{15}    & 3.6                            & 0.4                            & \multicolumn{1}{c|}{0.3}       & 12.012                 & \multicolumn{1}{c|}{12.017} & 2.696   & \multicolumn{1}{c|}{2.699}   & 0.471  & \multicolumn{1}{c|}{0.472}  & 0.331  & \multicolumn{1}{c|}{0.330}  & 0.232          & 0.231         \\ \hline
			\multicolumn{1}{c|}{16}    & 1.8                            & 0.4                            & \multicolumn{1}{c|}{-0.3}      & 2.166                  & \multicolumn{1}{c|}{2.302}  & 2.034   & \multicolumn{1}{c|}{1.984}   & 0.331  & \multicolumn{1}{c|}{0.215}  & 0.033  & \multicolumn{1}{c|}{0.070}  & 0.005          & 0.023         \\
			\multicolumn{1}{c|}{17}    & 5.4                            & 0.4                            & \multicolumn{1}{c|}{-0.3}      & 6.011                  & \multicolumn{1}{c|}{6.573}  & 3.691   & \multicolumn{1}{c|}{4.625}   & 0.358  & \multicolumn{1}{c|}{0.432}  & 0.037  & \multicolumn{1}{c|}{0.190}  & 0.003          & 0.084         \\
			\multicolumn{1}{c|}{18}    & 10.8                           & 0.4                            & \multicolumn{1}{c|}{-0.3}      & 12.005                 & \multicolumn{1}{c|}{12.969} & 2.349   & \multicolumn{1}{c|}{2.085}   & 0.358  & \multicolumn{1}{c|}{0.334}  & 0.035  & \multicolumn{1}{c|}{0.112}  & 0.004          & 0.037         \\ \hline
			\multicolumn{1}{c|}{19}    & 2.2                            & -0.4                           & \multicolumn{1}{c|}{0.3}       & 2.156                  & \multicolumn{1}{c|}{1.903}  & 1.954   & \multicolumn{1}{c|}{1.703}   & -0.289 & \multicolumn{1}{c|}{-0.300} & 0.018  & \multicolumn{1}{c|}{0.030}  & -0.001         & 0.000         \\
			\multicolumn{1}{c|}{20}    & 6.6                            & -0.4                           & \multicolumn{1}{c|}{0.3}       & 6.026                  & \multicolumn{1}{c|}{5.819}  & 3.625   & \multicolumn{1}{c|}{3.399}   & -0.342 & \multicolumn{1}{c|}{-0.300} & 0.031  & \multicolumn{1}{c|}{0.030}  & -0.002         & 0.000         \\
			\multicolumn{1}{c|}{21}    & 13.2                           & -0.4                           & \multicolumn{1}{c|}{0.3}       & 11.997                 & \multicolumn{1}{c|}{11.645} & 2.354   & \multicolumn{1}{c|}{2.550}   & -0.357 & \multicolumn{1}{c|}{-0.300} & 0.036  & \multicolumn{1}{c|}{0.030}  & -0.005         & 0.000         \\ \hline
			\multicolumn{1}{c|}{22}    & 3.4                            & -0.4                           & \multicolumn{1}{c|}{-0.3}      & 2.113                  & \multicolumn{1}{c|}{1.362}  & 2.015   & \multicolumn{1}{c|}{1.829}   & -0.325 & \multicolumn{1}{c|}{-0.304} & 0.189  & \multicolumn{1}{c|}{0.200}  & -0.105         & 0.000         \\
			\multicolumn{1}{c|}{23}    & 10.2                           & -0.4                           & \multicolumn{1}{c|}{-0.3}      & 6.035                  & \multicolumn{1}{c|}{6.610}  & 4.051   & \multicolumn{1}{c|}{3.982}   & -0.418 & \multicolumn{1}{c|}{-0.400} & 0.283  & \multicolumn{1}{c|}{0.300}  & -0.186         & -0.200        \\
			\multicolumn{1}{c|}{24}    & 20.4                           & -0.4                           & \multicolumn{1}{c|}{-0.3}      & 12.669                 & \multicolumn{1}{c|}{12.429} & 3.396   & \multicolumn{1}{c|}{4.297}   & 0.358  & \multicolumn{1}{c|}{0.312}  & 0.035  & \multicolumn{1}{c|}{0.097}  & 0.003          & 0.030         \\ \hline
			\multicolumn{14}{c}{$c = 0.5$, $n=3$}                                                                                                                                                                                                                                                                                                        \\ \hline
			\multicolumn{1}{c|}{1}     & 0.6                            & 0.3                            & \multicolumn{1}{c|}{0.4}       & 2.054                  & \multicolumn{1}{c|}{1.998}  & 2.069   & \multicolumn{1}{c|}{2.078}   & 0.349  & \multicolumn{1}{c|}{0.360}  & 0.241  & \multicolumn{1}{c|}{0.252}  & 0.166          & 0.176         \\
			\multicolumn{1}{c|}{2}     & 1.8                            & 0.3                            & \multicolumn{1}{c|}{0.4}       & 6.016                  & \multicolumn{1}{c|}{6.000}  & 3.771   & \multicolumn{1}{c|}{3.750}   & 0.361  & \multicolumn{1}{c|}{0.360}  & 0.254  & \multicolumn{1}{c|}{0.252}  & 0.176          & 0.176         \\
			\multicolumn{1}{c|}{3}     & 3.6                            & 0.3                            & \multicolumn{1}{c|}{0.4}       & 11.999                 & \multicolumn{1}{c|}{12.000} & 2.390   & \multicolumn{1}{c|}{2.388}   & 0.360  & \multicolumn{1}{c|}{0.360}  & 0.253  & \multicolumn{1}{c|}{0.252}  & 0.176          & 0.176         \\ \hline
			\multicolumn{1}{c|}{4}     & 2.2                            & 0.3                            & \multicolumn{1}{c|}{-0.4}      & 2.012                  & \multicolumn{1}{c|}{2.417}  & 1.877   & \multicolumn{1}{c|}{2.683}   & 0.265  & \multicolumn{1}{c|}{0.153}  & -0.023 & \multicolumn{1}{c|}{-0.023} & 0.002          & 0.004         \\
			\multicolumn{1}{c|}{5}     & 6.6                            & 0.3                            & \multicolumn{1}{c|}{-0.4}      & 6.005                  & \multicolumn{1}{c|}{5.132}  & 3.382   & \multicolumn{1}{c|}{2.828}   & 0.267  & \multicolumn{1}{c|}{0.325}  & -0.028 & \multicolumn{1}{c|}{-0.106} & 0.003          & 0.034         \\
			\multicolumn{1}{c|}{6}     & 13.2                           & 0.3                            & \multicolumn{1}{c|}{-0.4}      & 12.008                 & \multicolumn{1}{c|}{11.948} & 2.196   & \multicolumn{1}{c|}{2.363}   & 0.267  & \multicolumn{1}{c|}{0.153}  & -0.026 & \multicolumn{1}{c|}{-0.023} & 0.003          & 0.004         \\ \hline
			\multicolumn{1}{c|}{7}     & 1.8                            & -0.3                           & \multicolumn{1}{c|}{0.4}       & 2.022                  & \multicolumn{1}{c|}{1.785}  & 1.853   & \multicolumn{1}{c|}{1.581}   & -0.250 & \multicolumn{1}{c|}{-0.200} & -0.026 & \multicolumn{1}{c|}{0.000}  & -0.002         & 0.000         \\
			\multicolumn{1}{c|}{8}     & 5.4                            & -0.3                           & \multicolumn{1}{c|}{0.4}       & 6.000                  & \multicolumn{1}{c|}{5.599}  & 3.374   & \multicolumn{1}{c|}{1.751}   & -0.265 & \multicolumn{1}{c|}{-0.200} & -0.026 & \multicolumn{1}{c|}{0.000}  & -0.003         & 0.000         \\
			\multicolumn{1}{c|}{9}     & 10.8                           & -0.3                           & \multicolumn{1}{c|}{0.4}       & 11.999                 & \multicolumn{1}{c|}{11.216} & 2.201   & \multicolumn{1}{c|}{2.333}   & -0.266 & \multicolumn{1}{c|}{-0.200} & -0.027 & \multicolumn{1}{c|}{0.000}  & -0.003         & 0.000         \\ \hline
			\multicolumn{1}{c|}{10}    & 3.4                            & -0.3                           & \multicolumn{1}{c|}{-0.4}      & 2.019                  & \multicolumn{1}{c|}{2.181}  & 1.963   & \multicolumn{1}{c|}{2.133}   & -0.309 & \multicolumn{1}{c|}{-0.307} & 0.206  & \multicolumn{1}{c|}{0.201}  & -0.133         & 0.000         \\
			\multicolumn{1}{c|}{11}    & 10.2                           & -0.3                           & \multicolumn{1}{c|}{-0.4}      & 6.009                  & \multicolumn{1}{c|}{7.248}  & 3.705   & \multicolumn{1}{c|}{4.798}   & -0.347 & \multicolumn{1}{c|}{-0.300} & 0.241  & \multicolumn{1}{c|}{0.200}  & -0.167         & 0.000         \\
			\multicolumn{1}{c|}{12}    & 20.4                           & -0.3                           & \multicolumn{1}{c|}{-0.4}      & 11.994                 & \multicolumn{1}{c|}{12.153} & 2.394   & \multicolumn{1}{c|}{2.598}   & -0.361 & \multicolumn{1}{c|}{-0.300} & 0.252  & \multicolumn{1}{c|}{0.200}  & -0.176         & 0.000         \\ \hline
			\multicolumn{1}{c|}{13}    & 0.6                            & 0.4                            & \multicolumn{1}{c|}{0.3}       & 2.079                  & \multicolumn{1}{c|}{2.006}  & 2.427   & \multicolumn{1}{c|}{2.430}   & 0.462  & \multicolumn{1}{c|}{0.471}  & 0.320  & \multicolumn{1}{c|}{0.330}  & 0.219          & 0.231         \\
			\multicolumn{1}{c|}{14}    & 1.8                            & 0.4                            & \multicolumn{1}{c|}{0.3}       & 6.006                  & \multicolumn{1}{c|}{6.009}  & 2.772   & \multicolumn{1}{c|}{2.775}   & 0.472  & \multicolumn{1}{c|}{0.472}  & 0.329  & \multicolumn{1}{c|}{0.330}  & 0.230          & 0.231         \\
			\multicolumn{1}{c|}{15}    & 3.6                            & 0.4                            & \multicolumn{1}{c|}{0.3}       & 11.997                 & \multicolumn{1}{c|}{12.000} & 2.701   & \multicolumn{1}{c|}{2.698}   & 0.472  & \multicolumn{1}{c|}{0.472}  & 0.329  & \multicolumn{1}{c|}{0.330}  & 0.229          & 0.231         \\ \hline
			\multicolumn{1}{c|}{16}    & 1.8                            & 0.4                            & \multicolumn{1}{c|}{-0.3}      & 2.018                  & \multicolumn{1}{c|}{2.395}  & 2.040   & \multicolumn{1}{c|}{2.509}   & 0.352  & \multicolumn{1}{c|}{0.509}  & 0.036  & \multicolumn{1}{c|}{0.259}  & 0.003          & 0.132         \\
			\multicolumn{1}{c|}{17}    & 5.4                            & 0.4                            & \multicolumn{1}{c|}{-0.3}      & 5.998                  & \multicolumn{1}{c|}{6.002}  & 3.680   & \multicolumn{1}{c|}{3.766}   & 0.359  & \multicolumn{1}{c|}{0.372}  & 0.037  & \multicolumn{1}{c|}{0.139}  & 0.003          & 0.052         \\
			\multicolumn{1}{c|}{18}    & 10.8                           & 0.4                            & \multicolumn{1}{c|}{-0.3}      & 12.014                 & \multicolumn{1}{c|}{12.458} & 2.358   & \multicolumn{1}{c|}{2.360}   & 0.361  & \multicolumn{1}{c|}{0.305}  & 0.039  & \multicolumn{1}{c|}{0.093}  & 0.004          & 0.029         \\ \hline
			\multicolumn{1}{c|}{19}    & 2.2                            & -0.4                           & \multicolumn{1}{c|}{0.3}       & 2.034                  & \multicolumn{1}{c|}{1.923}  & 1.967   & \multicolumn{1}{c|}{1.809}   & -0.318 & \multicolumn{1}{c|}{-0.300} & 0.024  & \multicolumn{1}{c|}{0.000}  & -0.002         & 0.000         \\
			\multicolumn{1}{c|}{20}    & 6.6                            & -0.4                           & \multicolumn{1}{c|}{0.3}       & 6.014                  & \multicolumn{1}{c|}{6.804}  & 3.651   & \multicolumn{1}{c|}{2.328}   & -0.346 & \multicolumn{1}{c|}{-0.300} & 0.031  & \multicolumn{1}{c|}{0.000}  & -0.004         & 0.000         \\
			\multicolumn{1}{c|}{21}    & 13.2                           & -0.4                           & \multicolumn{1}{c|}{0.3}       & 12.004                 & \multicolumn{1}{c|}{11.879} & 2.356   & \multicolumn{1}{c|}{1.544}   & -0.358 & \multicolumn{1}{c|}{-0.317} & 0.034  & \multicolumn{1}{c|}{0.000}  & -0.004         & 0.000         \\ \hline
			\multicolumn{1}{c|}{22}    & 3.4                            & -0.4                           & \multicolumn{1}{c|}{-0.3}      & 2.018                  & \multicolumn{1}{c|}{1.903}  & 1.455   & \multicolumn{1}{c|}{1.449}   & -0.411 & \multicolumn{1}{c|}{-0.400} & 0.272  & \multicolumn{1}{c|}{0.300}  & -0.175         & -0.100        \\
			\multicolumn{1}{c|}{23}    & 10.2                           & -0.4                           & \multicolumn{1}{c|}{-0.3}      & 6.027                  & \multicolumn{1}{c|}{9.649}  & 4.142   & \multicolumn{1}{c|}{8.265}   & -0.429 & \multicolumn{1}{c|}{-0.400} & 0.294  & \multicolumn{1}{c|}{0.300}  & -0.195         & -0.100        \\
			\multicolumn{1}{c|}{24}    & 20.4                           & -0.4                           & \multicolumn{1}{c|}{-0.3}      & 11.994                 & \multicolumn{1}{c|}{12.153} & 2.693   & \multicolumn{1}{c|}{6.598}   & -0.471 & \multicolumn{1}{c|}{-0.400} & 0.328  & \multicolumn{1}{c|}{0.300}  & -0.228         & -0.200        \\ \hline
	\end{tabular}}
	\label{tab1}
\end{table}
 The asymptotic properties of $\boldsymbol{\hat{\Theta}}$ so obtained can be derived by an approach discussed in \cite{zhusp}, and is omitted to avoid repetition.
\begin{table}[H]
	\center
	\caption{Comparison of moments from  sp INGARCH(1,1) models with those obtained via linear model formulae("lin"): mean, dispersion ratio and ACF at lags 1 -3, for the cases $c=0.25$ and $c\to0$.}
	\renewcommand{\arraystretch}{1.2}
	\scalebox{0.6}{%
		\begin{tabular}{cccccccccccccc}
			\multicolumn{1}{l|}{}      & \multicolumn{1}{l}{}           & \multicolumn{1}{l}{}           & \multicolumn{1}{l|}{}          & \multicolumn{2}{c|}{$\mu$}           & \multicolumn{2}{c|}{$\sigma ^2 / \mu$} & \multicolumn{2}{c|}{$\rho (1)$}      & \multicolumn{2}{c|}{$\rho (2)$}      & \multicolumn{2}{c}{$\rho (3)$} \\
			\multicolumn{1}{l|}{Model} & \multicolumn{1}{c}{$\alpha_0$} & \multicolumn{1}{c}{$\alpha_1$} & \multicolumn{1}{c|}{$\beta_1$} & sp     & \multicolumn{1}{c|}{lin}    & sp      & \multicolumn{1}{c|}{lin}     & sp     & \multicolumn{1}{c|}{lin}    & sp     & \multicolumn{1}{c|}{lin}    & sp             & lin           \\ \hline
			\multicolumn{14}{c}{$c=0.25$, $n=3$}                                                                                                                                                                                                                                                                                         \\ \hline
			\multicolumn{1}{c|}{1}     & 0.6                            & 0.3                            & \multicolumn{1}{c|}{0.4}       & 1.932  & \multicolumn{1}{c|}{1.932}  & 1.985   & \multicolumn{1}{c|}{1.926}   & 0.314  & \multicolumn{1}{c|}{0.300}  & 0.226  & \multicolumn{1}{c|}{0.250}  & 0.147          & 0.170         \\
			\multicolumn{1}{c|}{2}     & 1.8                            & 0.3                            & \multicolumn{1}{c|}{0.4}       & 6.083  & \multicolumn{1}{c|}{6.084}  & 2.160   & \multicolumn{1}{c|}{2.099}   & 0.367  & \multicolumn{1}{c|}{0.300}  & 0.256  & \multicolumn{1}{c|}{0.250}  & 0.171          & 0.170         \\
			\multicolumn{1}{c|}{3}     & 3.6                            & 0.3                            & \multicolumn{1}{c|}{0.4}       & 11.898 & \multicolumn{1}{c|}{11.898} & 3.705   & \multicolumn{1}{c|}{3.577}   & 0.344  & \multicolumn{1}{c|}{0.300}  & 0.255  & \multicolumn{1}{c|}{0.250}  & 0.172          & 0.170         \\ \hline
			\multicolumn{1}{c|}{4}     & 2.2                            & 0.3                            & \multicolumn{1}{c|}{-0.4}      & 2.045  & \multicolumn{1}{c|}{2.044}  & 1.921   & \multicolumn{1}{c|}{1.877}   & 0.268  & \multicolumn{1}{c|}{0.250}  & -0.042 & \multicolumn{1}{c|}{-0.044} & 0.005          & 0.004         \\
			\multicolumn{1}{c|}{5}     & 6.6                            & 0.3                            & \multicolumn{1}{c|}{-0.4}      & 6.036  & \multicolumn{1}{c|}{6.036}  & 3.335   & \multicolumn{1}{c|}{3.288}   & 0.264  & \multicolumn{1}{c|}{0.250}  & -0.045 & \multicolumn{1}{c|}{-0.044} & 0.010          & 0.004         \\
			\multicolumn{1}{c|}{6}     & 13.2                           & 0.3                            & \multicolumn{1}{c|}{-0.4}      & 12.009 & \multicolumn{1}{c|}{12.009} & 2.214   & \multicolumn{1}{c|}{2.215}   & 0.290  & \multicolumn{1}{c|}{0.250}  & -0.042 & \multicolumn{1}{c|}{-0.044} & 0.009          & 0.004         \\ \hline
			\multicolumn{1}{c|}{7}     & 1.8                            & -0.3                           & \multicolumn{1}{c|}{0.4}       & 2.019  & \multicolumn{1}{c|}{2.019}  & 1.894   & \multicolumn{1}{c|}{1.896}   & -0.263 & \multicolumn{1}{c|}{-0.220} & -0.029 & \multicolumn{1}{c|}{-0.030} & -0.011         & -0.010        \\
			\multicolumn{1}{c|}{8}     & 5.4                            & -0.3                           & \multicolumn{1}{c|}{0.4}       & 5.991  & \multicolumn{1}{c|}{5.991}  & 2.141   & \multicolumn{1}{c|}{2.129}   & -0.257 & \multicolumn{1}{c|}{-0.220} & -0.037 & \multicolumn{1}{c|}{-0.030} & -0.010         & -0.010        \\
			\multicolumn{1}{c|}{9}     & 10.8                           & -0.3                           & \multicolumn{1}{c|}{0.4}       & 12.002 & \multicolumn{1}{c|}{12.002} & 2.157   & \multicolumn{1}{c|}{2.164}   & -0.269 & \multicolumn{1}{c|}{-0.220} & -0.031 & \multicolumn{1}{c|}{-0.030} & -0.012         & -0.011        \\ \hline
			\multicolumn{1}{c|}{10}    & 3.4                            & -0.3                           & \multicolumn{1}{c|}{-0.4}      & 1.994  & \multicolumn{1}{c|}{1.993}  & 1.522   & \multicolumn{1}{c|}{1.534}   & -0.337 & \multicolumn{1}{c|}{-0.320} & 0.231  & \multicolumn{1}{c|}{0.200}  & -0.148         & -0.160        \\
			\multicolumn{1}{c|}{11}    & 10.2                           & -0.3                           & \multicolumn{1}{c|}{-0.4}      & 6.049  & \multicolumn{1}{c|}{6.048}  & 3.697   & \multicolumn{1}{c|}{3.684}   & -0.354 & \multicolumn{1}{c|}{-0.320} & 0.253  & \multicolumn{1}{c|}{0.200}  & -0.169         & -0.160        \\
			\multicolumn{1}{c|}{12}    & 20.4                           & -0.3                           & \multicolumn{1}{c|}{-0.4}      & 12.014 & \multicolumn{1}{c|}{12.014} & 2.412   & \multicolumn{1}{c|}{2.403}   & -0.378 & \multicolumn{1}{c|}{-0.320} & 0.257  & \multicolumn{1}{c|}{0.200}  & -0.189         & -0.160        \\ \hline
			\multicolumn{1}{c|}{13}    & 0.6                            & 0.4                            & \multicolumn{1}{c|}{0.3}       & 2.052  & \multicolumn{1}{c|}{2.052}  & 2.382   & \multicolumn{1}{c|}{2.243}   & 0.455  & \multicolumn{1}{c|}{0.450}  & 0.325  & \multicolumn{1}{c|}{0.290}  & 0.221          & 0.225         \\
			\multicolumn{1}{c|}{14}    & 1.8                            & 0.4                            & \multicolumn{1}{c|}{0.3}       & 5.825  & \multicolumn{1}{c|}{5.825}  & 4.008   & \multicolumn{1}{c|}{3.740}   & 0.426  & \multicolumn{1}{c|}{0.450}  & 0.291  & \multicolumn{1}{c|}{0.290}  & 0.202          & 0.225         \\
			\multicolumn{1}{c|}{15}    & 3.6                            & 0.4                            & \multicolumn{1}{c|}{0.3}       & 12.166 & \multicolumn{1}{c|}{12.166} & 2.760   & \multicolumn{1}{c|}{2.695}   & 0.478  & \multicolumn{1}{c|}{0.450}  & 0.337  & \multicolumn{1}{c|}{0.290}  & 0.240          & 0.225         \\ \hline
			\multicolumn{1}{c|}{16}    & 1.8                            & 0.4                            & \multicolumn{1}{c|}{-0.3}      & 1.980  & \multicolumn{1}{c|}{1.980}  & 1.957   & \multicolumn{1}{c|}{1.888}   & 0.355  & \multicolumn{1}{c|}{0.349}  & 0.032  & \multicolumn{1}{c|}{0.040}  & -0.006         & -0.005        \\
			\multicolumn{1}{c|}{17}    & 5.4                            & 0.4                            & \multicolumn{1}{c|}{-0.3}      & 6.090  & \multicolumn{1}{c|}{6.090}  & 3.847   & \multicolumn{1}{c|}{3.644}   & 0.385  & \multicolumn{1}{c|}{0.349}  & 0.052  & \multicolumn{1}{c|}{0.040}  & -0.009         & -0.005        \\
			\multicolumn{1}{c|}{18}    & 10.8                           & 0.4                            & \multicolumn{1}{c|}{-0.3}      & 11.983 & \multicolumn{1}{c|}{11.982} & 2.330   & \multicolumn{1}{c|}{2.290}   & 0.355  & \multicolumn{1}{c|}{0.349}  & 0.045  & \multicolumn{1}{c|}{0.040}  & -0.005         & -0.005        \\ \hline
			\multicolumn{1}{c|}{19}    & 2.2                            & -0.4                           & \multicolumn{1}{c|}{0.3}       & 3.108  & \multicolumn{1}{c|}{3.108}  & 2.088   & \multicolumn{1}{c|}{2.079}   & -0.300 & \multicolumn{1}{c|}{-0.325} & 0.032  & \multicolumn{1}{c|}{0.029}  & -0.003         & -0.003        \\
			\multicolumn{1}{c|}{20}    & 6.6                            & -0.4                           & \multicolumn{1}{c|}{0.3}       & 6.048  & \multicolumn{1}{c|}{6.048}  & 2.376   & \multicolumn{1}{c|}{2.408}   & -0.360 & \multicolumn{1}{c|}{-0.325} & 0.035  & \multicolumn{1}{c|}{0.029}  & -0.002         & -0.003        \\
			\multicolumn{1}{c|}{21}    & 13.2                           & -0.4                           & \multicolumn{1}{c|}{0.3}       & 12.051 & \multicolumn{1}{c|}{12.051} & 2.336   & \multicolumn{1}{c|}{2.346}   & -0.357 & \multicolumn{1}{c|}{-0.325} & 0.045  & \multicolumn{1}{c|}{0.029}  & -0.005         & -0.003        \\ \hline
			\multicolumn{1}{c|}{22}    & 3.4                            & -0.4                           & \multicolumn{1}{c|}{-0.3}      & 2.011  & \multicolumn{1}{c|}{2.011}  & 2.113   & \multicolumn{1}{c|}{2.176}   & -0.392 & \multicolumn{1}{c|}{-0.460} & 0.267  & \multicolumn{1}{c|}{0.250}  & -0.165         & -0.179        \\
			\multicolumn{1}{c|}{23}    & 10.2                           & -0.4                           & \multicolumn{1}{c|}{-0.3}      & 6.025  & \multicolumn{1}{c|}{6.025}  & 2.784   & \multicolumn{1}{c|}{2.834}   & -0.463 & \multicolumn{1}{c|}{-0.460} & 0.330  & \multicolumn{1}{c|}{0.250}  & -0.228         & -0.179        \\
			\multicolumn{1}{c|}{24}    & 20.4                           & -0.4                           & \multicolumn{1}{c|}{-0.3}      & 11.988 & \multicolumn{1}{c|}{11.988} & 2.623   & \multicolumn{1}{c|}{2.653}   & -0.473 & \multicolumn{1}{c|}{-0.460} & 0.321  & \multicolumn{1}{c|}{0.250}  & -0.216         & -0.179        \\ \hline
			\multicolumn{14}{c}{$c \to 0$, $n=3$}                                                                                                                                                                                                                                                                                        \\ \hline
			\multicolumn{1}{c|}{1}     & 0.6                            & 0.3                            & \multicolumn{1}{c|}{0.4}       & 1.961  & \multicolumn{1}{c|}{1.961}  & 2.004   & \multicolumn{1}{c|}{1.960}   & 0.332  & \multicolumn{1}{c|}{0.300}  & 0.216  & \multicolumn{1}{c|}{0.200}  & 0.159          & 0.140         \\
			\multicolumn{1}{c|}{2}     & 1.8                            & 0.3                            & \multicolumn{1}{c|}{0.4}       & 5.920  & \multicolumn{1}{c|}{5.921}  & 2.441   & \multicolumn{1}{c|}{2.379}   & 0.353  & \multicolumn{1}{c|}{0.300}  & 0.249  & \multicolumn{1}{c|}{0.200}  & 0.171          & 0.140         \\
			\multicolumn{1}{c|}{3}     & 3.6                            & 0.3                            & \multicolumn{1}{c|}{0.4}       & 12.130 & \multicolumn{1}{c|}{12.130} & 2.406   & \multicolumn{1}{c|}{2.377}   & 0.354  & \multicolumn{1}{c|}{0.300}  & 0.230  & \multicolumn{1}{c|}{0.200}  & 0.165          & 0.140         \\ \hline
			\multicolumn{1}{c|}{4}     & 2.2                            & 0.3                            & \multicolumn{1}{c|}{-0.4}      & 2.041  & \multicolumn{1}{c|}{2.040}  & 1.931   & \multicolumn{1}{c|}{1.890}   & 0.283  & \multicolumn{1}{c|}{0.290}  & -0.025 & \multicolumn{1}{c|}{-0.020} & 0.005          & 0.000         \\
			\multicolumn{1}{c|}{5}     & 6.6                            & 0.3                            & \multicolumn{1}{c|}{-0.4}      & 6.044  & \multicolumn{1}{c|}{6.044}  & 2.280   & \multicolumn{1}{c|}{2.241}   & 0.291  & \multicolumn{1}{c|}{0.290}  & -0.022 & \multicolumn{1}{c|}{-0.020} & 0.011          & 0.010         \\
			\multicolumn{1}{c|}{6}     & 13.2                           & 0.3                            & \multicolumn{1}{c|}{-0.4}      & 12.070 & \multicolumn{1}{c|}{12.070} & 2.202   & \multicolumn{1}{c|}{2.199}   & 0.262  & \multicolumn{1}{c|}{0.290}  & -0.027 & \multicolumn{1}{c|}{-0.020} & 0.011          & 0.010         \\ \hline
			\multicolumn{1}{c|}{7}     & 1.8                            & -0.3                           & \multicolumn{1}{c|}{0.4}       & 1.996  & \multicolumn{1}{c|}{1.994}  & 1.860   & \multicolumn{1}{c|}{1.878}   & -0.260 & \multicolumn{1}{c|}{-0.267} & -0.029 & \multicolumn{1}{c|}{-0.030} & -0.008         & -0.003        \\
			\multicolumn{1}{c|}{8}     & 5.4                            & -0.3                           & \multicolumn{1}{c|}{0.4}       & 6.026  & \multicolumn{1}{c|}{7.124}  & 2.826   & \multicolumn{1}{c|}{2.833}   & -0.266 & \multicolumn{1}{c|}{-0.267} & -0.026 & \multicolumn{1}{c|}{-0.031} & -0.005         & -0.003        \\
			\multicolumn{1}{c|}{9}     & 10.8                           & -0.3                           & \multicolumn{1}{c|}{0.4}       & 12.022 & \multicolumn{1}{c|}{12.022} & 2.244   & \multicolumn{1}{c|}{2.238}   & -0.272 & \multicolumn{1}{c|}{-0.271} & -0.033 & \multicolumn{1}{c|}{-0.035} & -0.008         & -0.003        \\ \hline
			\multicolumn{1}{c|}{10}    & 3.4                            & -0.3                           & \multicolumn{1}{c|}{-0.4}      & 2.027  & \multicolumn{1}{c|}{14.276} & 2.472   & \multicolumn{1}{c|}{2.474}   & -0.319 & \multicolumn{1}{c|}{-0.500} & 0.199  & \multicolumn{1}{c|}{0.243}  & -0.142         & -0.140        \\
			\multicolumn{1}{c|}{11}    & 10.2                           & -0.3                           & \multicolumn{1}{c|}{-0.4}      & 6.062  & \multicolumn{1}{c|}{9.880}  & 3.837   & \multicolumn{1}{c|}{3.843}   & -0.360 & \multicolumn{1}{c|}{-0.500} & 0.269  & \multicolumn{1}{c|}{0.243}  & -0.145         & -0.140        \\
			\multicolumn{1}{c|}{12}    & 20.4                           & -0.3                           & \multicolumn{1}{c|}{-0.4}      & 11.975 & \multicolumn{1}{c|}{11.977} & 2.395   & \multicolumn{1}{c|}{2.396}   & -0.360 & \multicolumn{1}{c|}{-0.361} & 0.242  & \multicolumn{1}{c|}{0.243}  & -0.141         & -0.140        \\ \hline
			\multicolumn{1}{c|}{13}    & 0.6                            & 0.4                            & \multicolumn{1}{c|}{0.3}       & 1.951  & \multicolumn{1}{c|}{1.955}  & 2.267   & \multicolumn{1}{c|}{2.260}   & 0.438  & \multicolumn{1}{c|}{0.440}  & 0.282  & \multicolumn{1}{c|}{0.280}  & 0.188          & 0.200         \\
			\multicolumn{1}{c|}{14}    & 1.8                            & 0.4                            & \multicolumn{1}{c|}{0.3}       & 5.939  & \multicolumn{1}{c|}{6.218}  & 4.377   & \multicolumn{1}{c|}{4.299}   & 0.458  & \multicolumn{1}{c|}{0.454}  & 0.318  & \multicolumn{1}{c|}{0.324}  & 0.223          & 0.206         \\
			\multicolumn{1}{c|}{15}    & 3.6                            & 0.4                            & \multicolumn{1}{c|}{0.3}       & 12.060 & \multicolumn{1}{c|}{12.056} & 2.689   & \multicolumn{1}{c|}{2.676}   & 0.473  & \multicolumn{1}{c|}{0.466}  & 0.317  & \multicolumn{1}{c|}{0.311}  & 0.222          & 0.207         \\ \hline
			\multicolumn{1}{c|}{16}    & 1.8                            & 0.4                            & \multicolumn{1}{c|}{-0.3}      & 2.036  & \multicolumn{1}{c|}{2.040}  & 2.057   & \multicolumn{1}{c|}{2.061}   & 0.361  & \multicolumn{1}{c|}{0.362}  & 0.041  & \multicolumn{1}{c|}{0.040}  & 0.011          & 0.010         \\
			\multicolumn{1}{c|}{17}    & 5.4                            & 0.4                            & \multicolumn{1}{c|}{-0.3}      & 6.032  & \multicolumn{1}{c|}{6.034}  & 3.747   & \multicolumn{1}{c|}{3.723}   & 0.361  & \multicolumn{1}{c|}{0.361}  & 0.028  & \multicolumn{1}{c|}{0.031}  & 0.006          & 0.003         \\
			\multicolumn{1}{c|}{18}    & 10.8                           & 0.4                            & \multicolumn{1}{c|}{-0.3}      & 11.915 & \multicolumn{1}{c|}{12.280} & 2.363   & \multicolumn{1}{c|}{2.362}   & 0.350  & \multicolumn{1}{c|}{0.259}  & 0.042  & \multicolumn{1}{c|}{0.046}  & 0.012          & 0.012         \\ \hline
			\multicolumn{1}{c|}{19}    & 2.2                            & -0.4                           & \multicolumn{1}{c|}{0.3}       & 1.996  & \multicolumn{1}{c|}{1.987}  & 1.962   & \multicolumn{1}{c|}{2.017}   & -0.329 & \multicolumn{1}{c|}{-0.354} & 0.026  & \multicolumn{1}{c|}{0.039}  & -0.005         & -0.004        \\
			\multicolumn{1}{c|}{20}    & 6.6                            & -0.4                           & \multicolumn{1}{c|}{0.3}       & 5.970  & \multicolumn{1}{c|}{5.970}  & 2.341   & \multicolumn{1}{c|}{2.342}   & -0.347 & \multicolumn{1}{c|}{-0.352} & 0.034  & \multicolumn{1}{c|}{0.036}  & -0.007         & -0.004        \\
			\multicolumn{1}{c|}{21}    & 13.2                           & -0.4                           & \multicolumn{1}{c|}{0.3}       & 12.031 & \multicolumn{1}{c|}{11.866} & 2.455   & \multicolumn{1}{c|}{2.445}   & -0.370 & \multicolumn{1}{c|}{-0.281} & 0.057  & \multicolumn{1}{c|}{0.079}  & -0.010         & -0.004        \\ \hline
			\multicolumn{1}{c|}{22}    & 3.4                            & -0.4                           & \multicolumn{1}{c|}{-0.3}      & 2.015  & \multicolumn{1}{c|}{2.000}  & 2.310   & \multicolumn{1}{c|}{2.300}   & -0.405 & \multicolumn{1}{c|}{-0.435} & 0.292  & \multicolumn{1}{c|}{0.325}  & -0.179         & -0.166        \\
			\multicolumn{1}{c|}{23}    & 10.2                           & -0.4                           & \multicolumn{1}{c|}{-0.3}      & 6.000  & \multicolumn{1}{c|}{6.571}  & 2.756   & \multicolumn{1}{c|}{2.749}   & -0.456 & \multicolumn{1}{c|}{-0.490} & 0.306  & \multicolumn{1}{c|}{0.308}  & -0.222         & -0.220        \\
			\multicolumn{1}{c|}{24}    & 20.4                           & -0.4                           & \multicolumn{1}{c|}{-0.3}      & 12.025 & \multicolumn{1}{c|}{12.027} & 2.654   & \multicolumn{1}{c|}{2.686}   & -0.468 & \multicolumn{1}{c|}{-0.476} & 0.327  & \multicolumn{1}{c|}{0.328}  & -0.228         & -0.226        \\ \hline
	\end{tabular}}
	\label{tab2}
\end{table}
In the subsequent section, we briefly discuss how neural network architecture is embedded into the INGARCH framework and define the neural NB- INGARCH model. 
\section{Brief overview of neural INGARCH models}\label{ovr}

The INGARCH(p,q) model in (\ref{nbing}) can be reformulated to a general expression:
\begin{equation}
	\label{eqlam}
	E[X_t | \mathscr{F}_{t-1}] = \lambda_t = g(X_{t-1},\ldots,X_{t-p},\lambda_{t-1},\ldots,\lambda_{t-p}),
\end{equation}
where $g$ is the response function of choice, $\mathscr{F}_{t-1}$ reprises the definition in (\ref{sppois}), and the conditional distribution of $X_t$ is a suitable discrete distribution with mean $\lambda_t$. Defining $ g(x) = x $ simplifies (\ref{eqlam})  to the earlier representation (\ref{nbing}).  \cite{jahn2023artificial}  considered response functions from the class of artificial neural networks. In particular, $g$ is assumed to be a single hidden layer feedforward network (SLFN). Let $g$ be defined as
\begin{equation}
	\label{eqann}
	g^{\textrm{ANN}}(u^0,u^1,\mathbf{x}) = f_1\left(\sum_{l=1}^{L}u_l^1.f_0\left(\sum_{k=1}^K u_{kl}^0x_k\right)\right).
\end{equation}

The SLFN function relies on a set of parameters $ u_l^1$ and $u_{kl}^0$. The input vector  $\mathbf{x}$ includes a constant and the respective lagged values, given by $\mathbf{x} = (1, X_{t-1}, \ldots, X_{t-p}, \lambda_{t-1}, \ldots, \lambda_{t-q} )$ with $K = p + q + 1$ elements. The index $l$ $( 1 \leq l \leq L )$ refers to the hidden neurons, where $L$ indicates the model's complexity. In practice, the optimal number of hidden neurons $L$ can be determined using information criteria. The activation functions $f_0$ and $f_1$ are required to be increasing and continuously differentiable. Due to the nonlinearity of $f_0$ and $f_1$, the neural INGARCH model can be essentially interpreted as: The conditional expectation of $X_t$ is a nonlinear combination of the nonlinear combinations of $X_{t-1}, \ldots, X_{t-p}, \lambda_{t-1}, \ldots, \lambda_{t-q}$. \par
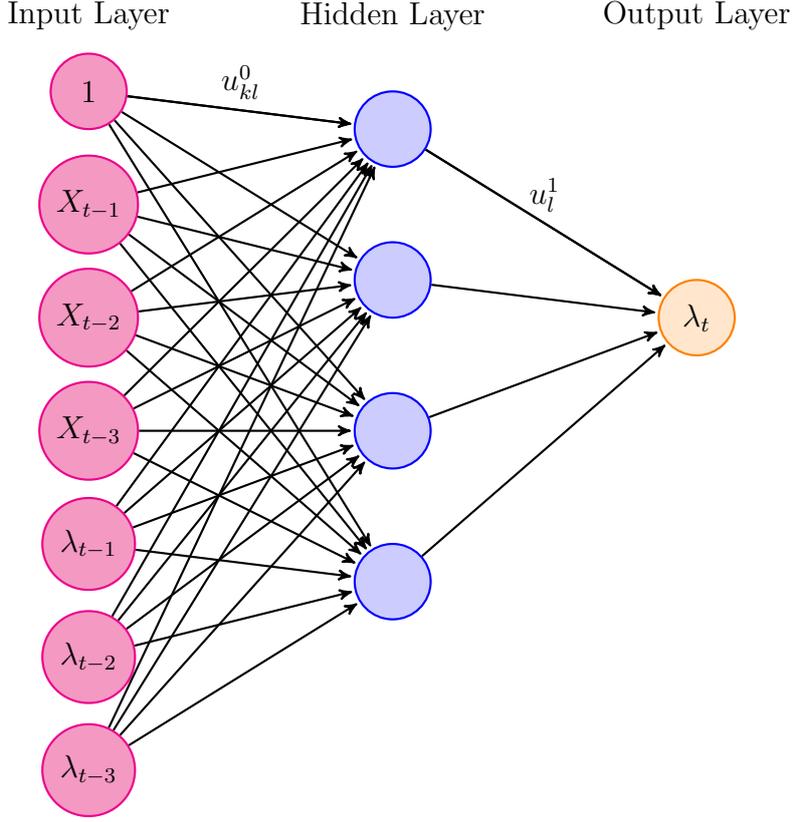
\begin{figure}[ht]
	\centering
	\begin{tikzpicture}[->, >=stealth', shorten >=1pt, auto, node distance=2cm, 
		thick, main node/.style={circle, draw, minimum size=1cm}]
		
		% Input layer heading
		\node at (0,6) {Input Layer};
		
		% Input layer nodes
		\node[main node, fill=magenta!50, draw=magenta] (I1) at (0,5) {$1$};
		\node[main node, fill=magenta!50, draw=magenta] (I2) at (0,3.5) {$X_{t-1}$};
		\node[main node, fill=magenta!50, draw=magenta] (I3) at (0,2) {$X_{t-2}$};
		\node[main node, fill=magenta!50, draw=magenta] (I4) at (0,0.5) {$X_{t-3}$};
		\node[main node, fill=magenta!50, draw=magenta] (I5) at (0,-1) {$\lambda_{t-1}$};
		\node[main node, fill=magenta!50, draw=magenta] (I6) at (0,-2.5) {$\lambda_{t-2}$};
		\node[main node, fill=magenta!50, draw=magenta] (I7) at (0,-4) {$\lambda_{t-3}$};
		
		% Hidden layer heading
		\node at (4,6) {Hidden Layer};
		
		% Hidden layer nodes
		\node[main node, fill=blue!20, draw=blue] (H1) at (4,4.5){} ;
		\node[main node, fill=blue!20, draw=blue] (H2) at (4,2.5){} ;
		\node[main node, fill=blue!20, draw=blue] (H3) at (4,0.5) {};
		\node[main node, fill=blue!20, draw=blue] (H4) at (4,-1.5) {};
		
		% Output layer heading
		\node at (8,6) {Output Layer};
		
		% Output layer node
		\node[main node, fill=orange!20, draw=orange] (O) at (8,2) {$\lambda_t$};
		
		% Connections from input to hidden layer
		\foreach \i in {1,2,3,4,5,6,7} {
			\foreach \h in {1,2,3,4} {
				\draw (I\i) -- (H\h);
			}
		}
		
		% Add label to the first connection from X_1 to H_1
		\draw (I1) -- node[midway, above] {$u_{kl}^0$} (H1);
		
		% Connections from hidden to output layer
		\foreach \h in {1,2,3,4} {
			\draw (H\h) -- (O);
		}
		
		% Add label to the first connection from H_1 to Y_t
		\draw (H1) -- node[midway, above] {$ u_l^1$} (O);
	\end{tikzpicture}
	\caption{Neural network architecture of neural INGARCH(3,3) model with four hidden neurons.}
	\label{fig:ann}
\end{figure}
\textcolor{blue}{Figure}  \ref{fig:ann} illustrates a specific configuration of a neural INGARCH(3,3) model, highlighting the concept of hidden neurons and activation functions. The SLFN consists of three layers: an input layer, a hidden layer, and an output layer. The neurons (depicted as circles) in the input layer correspond to the regressor variables. The hidden layer, which is unobserved, is exemplified with four hidden neurons. The single neuron in the output layer represents the conditional mean of the current value of the time series. Alternatively, an illustration could focus on the target value \(X_t\). As indicated by the direction of the arrows, information flows only in the forward direction from the input layer through the hidden layer to the output layer, which characterizes the network as a single hidden layer ``feedforward" network.\par
The trick lies in the choice of the activation functions $f_0 $ and $f_1$. The term ``activation'' refers to the biological analogy: a neuron transmits information based on its activation level, which is determined by the preceding layer. For the activation of the hidden layer, \cite{jahn2023artificial} used the logistic function, introducing the desired non-linearity into the model, thus $f_0(x) = \frac{1}{1 + \exp(-x)}$. This implies that the values of the hidden neurons range from 0 (completely inactive) to 1 (completely active). The activation of the output layer requires more careful consideration, with the choice of $f_1$ depending on the data. The softplus function, with $c=1$, i.e., $f_1(x) = \ln(1 + \exp(x))$ seemed suitable based on the requirements. \par

For a conditional Poisson distribution, the conditional log-likelihood function of an INGARCH model with response function $g$  (denoted as neu - PINGARCH(p,q)) and parameter vector $\mathbf{u}$ is
\begin{equation}
	l(\mathbf{u}) = \sum_{t} X_t \log g(\mathbf{u}, \mathbf{x}) - g(\mathbf{u}, \mathbf{x}) - log(X_t!), 
\end{equation}
which needs to be maximized with respect to $\mathbf{u}$. For practical modeling, it is noteworthy that the numerical maximization speed can be significantly improved in the case of a purely autoregressive model $ q = 0$. On the other hand, if $q > 0$, the log-likelihood function (and the gradient) are defined recursively, implying that an optimization algorithm needs to loop through the individual time periods in each iteration. According to the chain rule, we have the following relationship:

\begin{equation}
	\frac{\partial l}{\partial \mathbf{u}} = \frac{\partial l}{\partial g} \frac{\partial g}{\partial \mathbf{u}}.
\end{equation}

Regarding the gradient, the outer derivative is simply the derivative of the conditional Poisson log-likelihood with respect to the Poisson parameter $g(\mathbf{u}, \mathbf{x})$:
\begin{equation}
	\frac{\partial l}{\partial g} = \sum_t \left( \frac{X_t}{g(\mathbf{u}, \mathbf{x})} - 1 \right).
\end{equation}

The partial derivatives $ \frac{\partial g}{\partial \mathbf{u}}$ can be computed via the chain rule. In the case of the ANN response function $g = g_{\text{ANN}}$, this can be done efficiently using the backpropagation procedure. The only difference from conventional backpropagation is that the error which is back-propagated is now the relative error instead of the plain error, which arises from the first-order condition of the least squares approach. For efficient computation, it is advisable to use activation functions with simple derivatives. In this example, for the softplus function, $ f_1' = f_0 $ (i.e., the derivative of the softplus function is the logistic function), and for the logistic function, $f_0' = f_0 (1 - f_0)$. In what follows, we introduce the neural negative binomial INGARCH model.
\subsection{A neural negative binomial INGARCH model}
The neural negative binomial INGARCH model may be defined by
\begin{equation}
	X_t \mid \mathscr{F}_{t-1} \sim  NB(n,p_t),  \quad
	\lambda_t=\frac{n(1-p_t)}{p_t},
\end{equation}
where $\lambda_t = g(X_{t-1},\ldots,X_{t-p},\lambda_{t-1},\ldots,\lambda_{t-q})$ satisfies
(\ref{eqann}), with $f_0$ being chosen as the logistic function and $f_1$ as the softplus function with $c=1$ and the rest of the parameters satisfying conditions mentioned earlier. We denote the model as neu - NB-INGARCH(p,q). Thus, the conditional log likelihood reduces to
\begin{equation}
	\label{likl}
	l(\mathbf{u}) = \sum_{t=1}^{n} \left\{x_tlog\left(\frac{g(\mathbf{u,x})}{n}\right) - (n+x_t)log\left(1+\frac{g(\mathbf{u,x})}{n}\right) + e\right\},
\end{equation}
where $e=\sum_{v=1}^{x_t}log(v+n-1) - log(x_t!)$. The CMLEs of $\mathbf{u}$ can be obtained by maximizing (\ref{likl}), which is implemented using numerical methods. Details on model selection and diagnostics are given in \cite{jahn2023artificial}.

In the next section, we discuss a simulation study to better understand the estimation procedure in the case of softplus negative binomial INGARCH model.
\section{Simulation Study}
In this section, a numerical simulation study is performed to validate the conditional maximum likelihood estimation method in the case of softplus NB-INGARCH model discussed in \textcolor{blue}{Section} \ref{spn}. We generate samples sized $100, 500$, and $1000$ from sp NB-INGARCH (1,1) for 16 parameter combinations. For each configuration, we repeat 1000 iterations, calculate the average, and then calculate the average absolute bias (Abs. Bias) (i.e., $ \frac{1}{1000}\sum_{k=1}^{1000} |\hat{\Theta}_i^{(k)}- \Theta_i^0|,$ with $\Theta_i^0$ as the true value of the $i^{th}$ parameter and $\hat{\Theta}_i^{(k)}$ as the corresponding CMLE for the $k^{th}$ replication) and mean square error (MSE) of the estimates as shown in \textcolor{blue}{Tables} \ref{tab3} and \ref{tab4} for the cases when $c=1$ and $c=0.5$ respectively.\par

From the simulation results, it can be seen that the absolute bias and MSE decreases with increase in sample size. In particular, when \( \alpha_0 \), the estimates provided by specifying $c=0.5$ are more accurate and close to the true values, and the standard error of the estimate of \( \alpha_0 \) is smaller. Interestingly, for many cases,  even the estimates obtained for samples sized $100$ are very close to the actual values of the parameters.

A similar study is not applicable to neural INGARCH models for a couple of reasons. Firstly, the number of parameters in $\mathbf{u}$ may vary from sample to sample. Further, an individual estimation or optimization
need not always converge and so any assessment that requires
estimating the neural INGARCH model a large number of times should be avoided. In effect, the prediction accuracy has a more important role when it comes to a neural network model than the finite sample properties of the estimator as such models often cater to big data.
\begin{table}[H]
	\center
	\caption{CML estimation for simulated softplus INGARCH(1,1) processes: mean, absolute bias, and MSE of the estimates, for $c=1$.}
	\renewcommand{\arraystretch}{1.2}
	\scalebox{0.55}{%
\begin{tabular}{ccccccccccccccccc}
			\multicolumn{1}{l|}{}      & \multicolumn{1}{l}{} & \multicolumn{1}{l}{} & \multicolumn{1}{l}{}           & \multicolumn{1}{l|}{}     & \multicolumn{4}{c|}{Mean of}                                                                                                                               & \multicolumn{4}{c|}{Abs. Bias of}                                                                                                                  & \multicolumn{4}{c}{MSE of}                                                                                                                        \\
			\multicolumn{1}{l|}{Model} & $\alpha_0$           & $\alpha_1$           & \multicolumn{1}{c|}{$\beta_1$} & \multicolumn{1}{c|}{$s$}  & \multicolumn{1}{c}{$\hat{\alpha}_0$} & \multicolumn{1}{c}{$\hat{\alpha}_1$} & \multicolumn{1}{c}{$\hat{\beta}_1$} & \multicolumn{1}{c|}{$\hat{n}$} & \multicolumn{1}{c}{$\hat{\alpha}_0$} & \multicolumn{1}{c}{$\hat{\alpha}_1$} & \multicolumn{1}{c}{$\hat{\beta}_1$} & \multicolumn{1}{c|}{$\hat{n}$} & \multicolumn{1}{c}{$\hat{\alpha}_0$} & \multicolumn{1}{c}{$\hat{\alpha}_1$} & \multicolumn{1}{c}{$\hat{\beta}_1$} & \multicolumn{1}{c}{$\hat{n}$} \\ \hline
			\multicolumn{17}{c}{$\mu=2.5$, $n=3$}                                                                                                                                                                                                                                                                                                                                                                                                                                                                                                                                                                       \\
\hline
	\end{tabular}}
	\label{tab3}
\end{table}
\begin{table}[H]
	\center
	\caption{CML estimation for simulated softplus INGARCH(1,1) processes: mean, absolute bias, and MSE of the estimates, for $c=0.5$.}
	\renewcommand{\arraystretch}{1.2}
	\scalebox{0.6}{%
		\begin{tabular}{ccccccccccccccccc}
			\multicolumn{1}{l|}{}      & \multicolumn{1}{l}{} & \multicolumn{1}{l}{} & \multicolumn{1}{l}{}           & \multicolumn{1}{l|}{}     & \multicolumn{4}{c|}{Mean of}                                                                                                                               & \multicolumn{4}{c|}{Abs. Bias of}                                                                                                                  & \multicolumn{4}{c}{MSE of}                                                                                                                        \\
			\multicolumn{1}{l|}{Model} & $\alpha_0$           & $\alpha_1$           & \multicolumn{1}{c|}{$\beta_1$} & \multicolumn{1}{c|}{$s$}  & \multicolumn{1}{c}{$\hat{\alpha}_0$} & \multicolumn{1}{c}{$\hat{\alpha}_1$} & \multicolumn{1}{c}{$\hat{\beta}_1$} & \multicolumn{1}{c|}{$\hat{n}$} & \multicolumn{1}{c}{$\hat{\alpha}_0$} & \multicolumn{1}{c}{$\hat{\alpha}_1$} & \multicolumn{1}{c}{$\hat{\beta}_1$} & \multicolumn{1}{c|}{$\hat{n}$} & \multicolumn{1}{c}{$\hat{\alpha}_0$} & \multicolumn{1}{c}{$\hat{\alpha}_1$} & \multicolumn{1}{c}{$\hat{\beta}_1$} & \multicolumn{1}{c}{$\hat{n}$} \\ \hline
			\multicolumn{17}{c}{$\mu=2.5$, $n=3$}                                                                                                                                                                                                                                                                                                                                                                                                                                                                                                                                                                       \\
\hline
	\end{tabular}}
	\label{tab4}
\end{table}
In the following section, we proceed to analyse two real data sets to assert the applicability of the softplus and neural NB - INGARCH models.
\section{Real Data Applications}
In this section, we demonstrate the real-world applications of the softplus negative binomial INGARCH and neural negative binomial INGARCH models. The first study showcases the application of sp NB-INGARCH model and compares its performance with softplus Poisson INGARCH model. The second data analysis deals with a non-stationary healthcare data and elucidates how a neu-NB-INGARCH model can handle such a scenario through a hybrid framework.
\subsection{Syphilis Count Data}
The data set consists of weekly number of syphilis cases reported in the West South Central states of Arkansas, Lousiana, Oklahoma and Texas of the United States from 2007 to 2010. This data set is available in the R package ZIM. \textcolor{blue}{Figure} \ref{syphfig} shows the basic structure of the data such as time series plot, ACF, and partial ACF (PACF) plots. The data is overdispersed, with an average value of $12.646$ and a dispersion (ratio of sample variance to mean) of $9.331$, so we choose softplus negattive binomial analysis as the appropriate inference method. We fit the sp PINGARCH (1,0), sp PINGARCH(2,0) and sp PINGARCH(1,1) models with Poisson conditional distribution, and sp NB-INGARCH (1,0), sp NB-INGARCH(2,0) and sp NB-INGARCH(1,1)  for comparison. Using the estimation method detailed in \cite{weiss2022softplus}(for sp PINGARCH models) and \textcolor{blue}{Subsection} \ref{est} of this article, we obtain the parameter estimates for the models, showing the corresponding Akaike information criterion (AIC) and Bayesian information criterion (BIC) values in \textcolor{blue}{Table} \ref{tab5}.

\begin{figure}[H]
	\centering
	\includegraphics[scale=0.5]{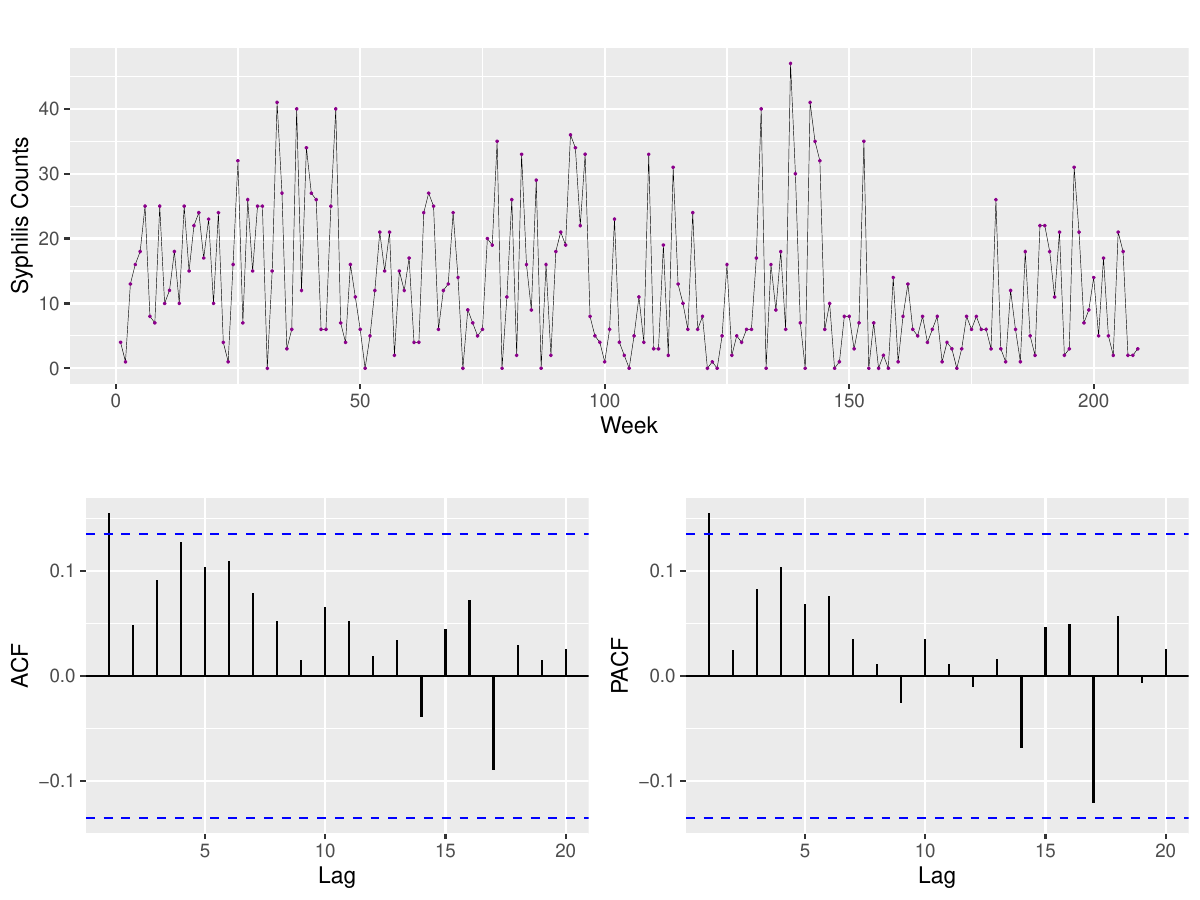}
	\caption{Time series, ACF and PACF plots for Syphilis data }
	\label{syphfig}
\end{figure}
As observed, the softplus NB-INGARCH models are superior to the sp PINGARCH models, yielding lower AIC and BIC values.  The least AIC and BIC value correspond to sp NB-INGARCH(2,0) model, and hence the model seems to provide a better fit than its counterparts chosen for comparison. The model adequacy diagnostics plots using ACF plots of residuals and cumulative periodograms for the sp PINGARCH and sp NB-INGARCH models are presented as \textcolor{blue}{Figures} \ref{fig7} and \ref{fig8} respectively in \textcolor{blue}{Appendix} \ref{diag}. One might be curious as to why a data with negative ACF was not chosen for analysis; as for the present data, the ACF at lag 1 is 0.155. This has already been dealt with by  \cite{weiss2022softplus} (Chemical process data, with ACF at lag 1 as -0.588), in which the sp NB-INGARCH(1,1) was found to be an appropriate fit.
\begin{table}[H]
	\centering
	\caption{Estimates (Standard errors in parantheses), AIC and BIC of softplus INGARCH  models fitted to Syphilis data.}
	\renewcommand{\arraystretch}{1.2}
	\scalebox{0.8}{%
		\begin{tabular}{lllllll}
			\hline
			\multirow{2}{*}{Model}                 & \multicolumn{4}{c}{Parameter}             & \multicolumn{1}{c}{\multirow{2}{*}{AIC}} & \multicolumn{1}{c}{\multirow{2}{*}{BIC}} \\ \cline{2-5}
			& 1        & 2        & 3        & 4        & \multicolumn{1}{c}{}                     & \multicolumn{1}{c}{}                     \\ \hline
			sp PINGARCH(1,0)                       & 10.6634  & 0.1595   &          &          & \multirow{2}{*}{2715.43}                 & \multirow{2}{*}{2722.11}                 \\
			($\alpha_0$, $\alpha_1$)               & (0.3694) & (0.0243) &          &          &                                          &                                          \\
			sp PINGARCH(2,0)                       & 10.4347  & 0.1514   & 0.0298   &          & \multirow{2}{*}{2697.95}                 & \multirow{2}{*}{2707.95}                 \\
			($\alpha_0$,$\alpha_1$,$\alpha_2$)     & (0.4487) & (0.0245) & (0.0233) &          &                                          &                                          \\
			sp PINGARCH(1,1)                       & 1.1202   & 0.1006   & 0.8102   &          & \multirow{2}{*}{2672.65}                 & \multirow{2}{*}{2682.66}                 \\
			($\alpha_0$,$\alpha_1$,$\beta_1$)      & (0.3185) & (0.0160) & (0.0348) &          &                                          &                                          \\ \hline
			sp NB-INGARCH(1,0)                     & 10.6054  & 0.1646   & 1.2224   &          & \multirow{2}{*}{1488.14}                 & \multirow{2}{*}{1498.15}                 \\
			($\alpha_0$, $\alpha_1$,$n$)           & (1.2123) & (0.0875) & (0.1326) &          &                                          &                                          \\
			sp NB-INGARCH(2,0)                     & 10.3475  & 0.1564   & 0.0324   & 1.2358   & \multirow{2}{*}{1484.47}                 & \multirow{2}{*}{1497.80}                 \\
			($\alpha_0$,$\alpha_1$,$\alpha_2$,$n$) & (1.4788) & (0.0877) & (0.0790) & (0.1346) &                                          &                                          \\
			sp NB-INGARCH(1,1)                     & 1.0118   & 0.1073   & 0.8125   & 1.2535   & \multirow{2}{*}{1485.40}                 & \multirow{2}{*}{1498.73}                 \\
			($\alpha_0$,$\alpha_1$,$\beta_1$,$n$)  & (0.9283) & (0.0552) & (0.1069) & (0.1369) &                                          &                                          \\ \hline
	\end{tabular}}
	\label{tab5}
\end{table}
In the next subsection, we proceed to present the motivation and analyse a healthcare data on Emergency Department (ED) visits using neural INGARCH models for various conditional distributions. 
\subsection{Emergency Department (ED) Arrivals Data}
The hospital emergency department (ED) is the fundamental unit providing immediate response to emergency health issues. It is an essential component of any health system, offering care for urgent and potentially serious pathological conditions that could result in death or require immediate diagnosis and treatment to alleviate pain. ED activities are both intense and highly diverse, ranging from life-threatening conditions like cardiac arrest to serious or potentially serious illnesses needing hospital-based diagnosis or treatment. EDs also handle less critical emergencies that may necessitate hospitalization for diagnosis and provide initial treatment and observation without necessarily requiring admission. (See \cite{reboredo2023forecasting})

Patient arrivals in emergency departments (EDs) are uneven over time. The distribution of arrivals may vary according to the days of the week, or even over months (See \cite{example9}, \cite{example20}), with demand for care fluctuating based on holiday periods (demographic movements), respiratory virus epidemics, climatic and atmospheric changes, and social events \cite{example20}. Managing surges, which is a key challenge for ensuring efficient ED management and functioning \cite{example7}, \cite{example16}, \cite{example21}, \cite{example26}, is closely tied to the timely provision of treatment. Therefore, ED and hospital resources need to be planned with flexibility to adapt to cyclical changes in service demand.

Beyond the quantitative aspects of patient arrivals, there is a significant qualitative impact. The diagnostic and therapeutic activities in EDs influence the subsequent health outcomes of admitted patients, including the length of stay, complications, and patient satisfaction. Patient satisfaction with healthcare services in general is strongly influenced by technical quality and, importantly, by the perceived quality of the ED, which shapes overall perceptions of hospital performance \cite{example11}.

To avoid congestion and ensure appropriate delivery of medical services, efficient ED management requires accurate forecasting of patient inflows (See \cite{example3}, \cite{example5}, \cite{example8}). However, forecasting is challenging due to the high variability and overdispersion in daily and seasonal patient arrivals \cite{example20}, \cite{example23}. Previous empirical research has explored the dynamics of arrivals, primarily using Poisson and negative binomial models with various extensions (See \cite{example3}, \cite{example28}, \cite{example34}, \cite{example35}). \cite{reboredo2023forecasting} showed that the negative binomial INGARCH model, proposed by \cite{xu2012model}, can provide an improved fit and more accurate forecasts of patient arrivals by leveraging historical arrival data and capturing the volatility dynamics in the arrival process. However, the suggested models cannot be applied in the case of extreme non-linear dependence or non-stationarity in the data. Hence, we proceed to illustrate the use of neural INGARCH models in such scenario.

The data for this study consists of hourly patient arrivals at one of the largest emergency department units in the UK (name of the establishment is not disclosed), spanning the period from April 2014 to February 2019. These data were obtained from the hospital's emergency department administrative database (See \cite{rostami}).
The data is available at \href{https://doi.org/10.5281/zenodo.7874721}{https://doi.org/10.5281/zenodo.7874721}. In light of the repetitive patterns observed throughout the data, as affirmed by \textcolor{blue}{Figures} \ref{fig1} and \ref{fig2}, for the present study, we consider the data of the period ranging from 1st December 2018 to 1st February 2019. The training set consists of the data upto 31st January 2019 (including  31st January 2019) and the test set considers data on  1st February 2019. The mean and dispersion ratio of the number of arrivals are 15.057 and 5.170 respectively, indicating overdispersion.\par
\begin{figure}[ht]
	\centering
	\includegraphics[scale=0.7]{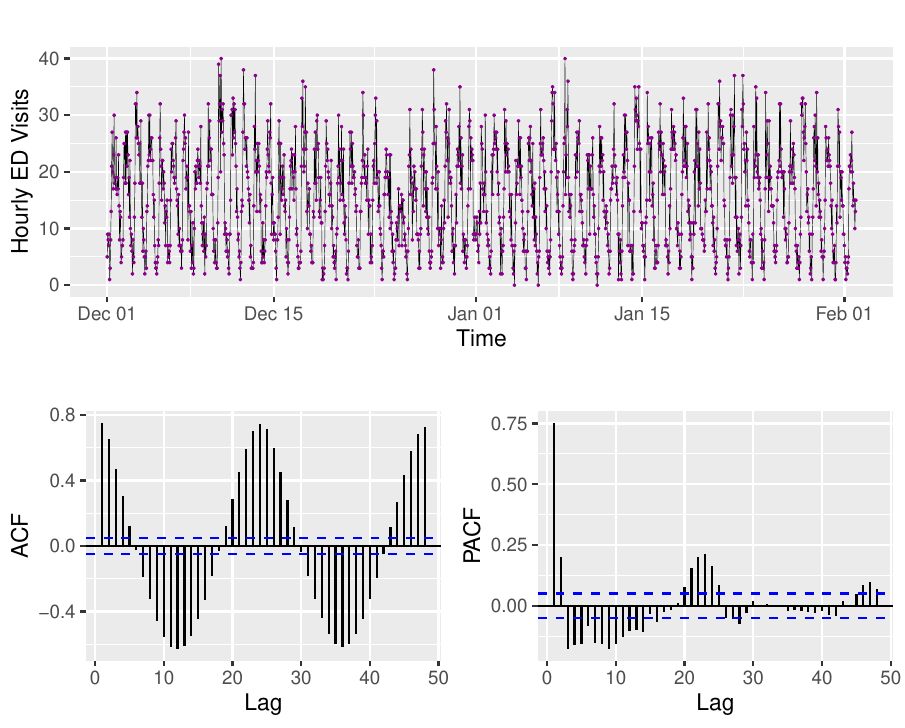}
	\caption{ Time series, ACF and PACF plots of ED arrivals.}
	\label{fig1}
\end{figure}
\textcolor{blue}{Figure} \ref{fig1} clearly illustrates that the inflow appears to follow a seasonal pattern with period 12 from the ACF plot. This expected pattern can be explained by the fact that the data is hourly. Thus, the data is non-stationary and cannot be at large handled by linear INGARCH models. Further, from {Figure} \ref{fig2} in \textcolor{blue}{Appendix} \ref{explore}, we can infer that the number of arrivals is higher from 9 am to 6 pm as compared to the other intervals.  \textcolor{blue}{Figure} \ref{fig3} shows that the arrival rate on Mondays and Saturdays are significantly higher than on other weekdays. \textcolor{blue}{Figure} \ref{fig4} depicts that the number of arrivals on a non-holiday are higher than that on a holiday.

\subsubsection{Model and Analysis}
The analysis is done in two steps like the hybrid variants discussed in \textcolor{blue}{Section} \ref{intro}. In the first step, a neural INGARCH model is fitted to the time series. For example, for a neural INGARCH(1,1) model:
\begin{equation}
	E[X_t\mid \mathscr{F}_{t-1}] = \lambda_t = 	g^{\textrm{ANN}}(u^0,u^1,\mathbf{x}) ,
\end{equation}
where $\mathbf{x} = (1,X_{t-1},\lambda_{t-1})$, and $g^{\textrm{ANN}}(.)$ is defined as in (\ref{eqann}) with $K=3$ for neural INGARCH(1,1). Models were fitted assuming Poisson and negative binomial distributions for $X_t\mid \mathscr{F}_{t-1}$.
However, the ACF of the pearson residuals ( denoted by $\{Z_t\}$) of the fitted models still showed seasonal behaviour as can be seen in \textcolor{blue}{Figures} \ref{fig5} and \ref{fig9}. Hence, to account for the seasonal pattern and other features of the data found in the exploratory analysis, a seasonal autoregressive integrated moving average with exogenous factors (SARIMAX) model (See \cite{aburto2007improved} and \cite{cools2009investigating})  is fitted to the residuals obtained from fitting the neural INGARCH models. 
\textcolor{blue}{Equation} (\ref{eqmodel}) presents the suggested form of SARIMAX model:

\begin{equation}
	\label{eqmodel}
	Z_t = \boldsymbol{\gamma}Y_{t} + \frac{\theta_q(B)\Theta_Q(B^S)}{\phi_p(B)\Phi_P(B^S)(1 - B)^d(1 - B^S)^D} \varepsilon_t, 
\end{equation}
where $Y= [I_{\textrm{Monday}}, I_{\textrm{Weekend}}, I_{\textrm{Winter}}]$, $I(.)$ denotes the indicator function corresponding to the factors whether the day of the observation is a Monday, Weekend, i.e., either Friday or Saturday, and/or is a holiday; $\boldsymbol{\gamma}$ denotes the corresponding vector of coefficients of the indicator variables,  $\theta_q(B)$ is the non-seasonal moving average operator with q-order, $\Theta_Q(B^S)$ the seasonal moving average operator with Q-order, $\phi_p(B)$ is the non-seasonal autoregressive operator with p-order,$\Phi_P(B^S)$ is the seasonal autoregressive operator with P-order, $(1 - B^S)^D$ is the seasonal differencing operator of order D,
$(1 - B)^d$ is the differencing operator of order d,  $S$ denotes Seasonal length (for the present hourly data $S=24$) and the residual error sequence $\{\varepsilon_t\}$ is assumed to be white noise.
We fit six models to the data -- neu PINGARCH(1,0)), neural neu - PINGARCH(2,0), neu PINGARCH (1,1), neu-NB-INGARCH(1,0), neu - NB- INGARCH(2,0) and neu-NB-INGARCH(1,1), based on the large significance at lag 1 displayed in the PACF plot of the data in \textcolor{blue}{Figure} \ref{fig1}. \textcolor{blue}{Figures} \ref{fig6} and \ref{fig10} display the ACF plots and cumulative periodograms for residuals on fitting the SARIMAX models. The best performing model is neu-NB-INGARCH(1,1) with SARIMAX(2,0,0)(2,1,1)[24] fitted to residuals, based on the minimum combined information criteria of both the consitutent models in \textcolor{blue}{Table} \ref{tab6}. The out-of-sample forecasting accuracy is computed for the four best performing models and is tabulated in \textcolor{blue}{Table} \ref{tab7}. The forecasts corresponding to the test set for neu-NB-INGARCH(1,1) hybrid model is plotted in \textcolor{blue}{Figure} \ref{fig11}.
	\begin{figure}[H]
	\centering
	\includegraphics[scale=0.5]{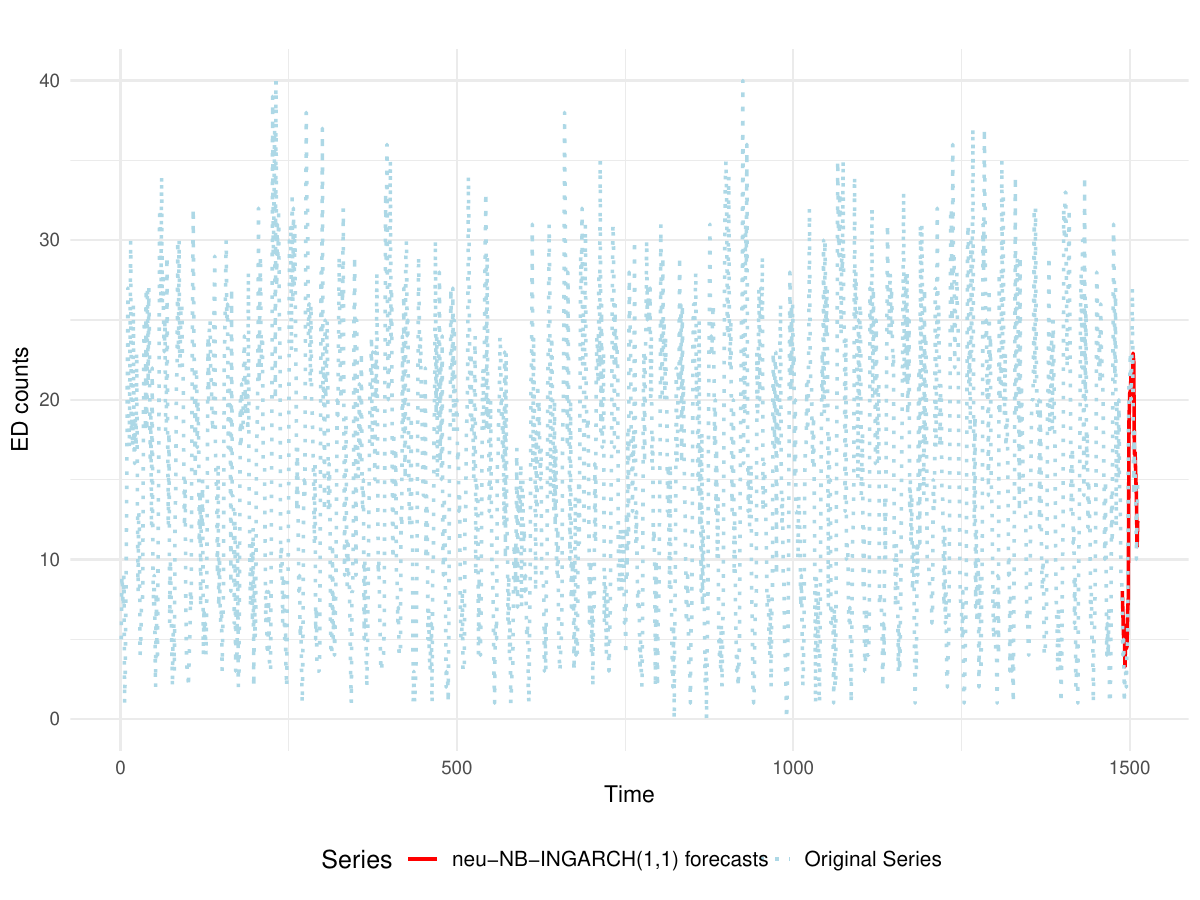}
	\caption{Time series plot of the original ED series along with neu-NB-INGARCH(1,1) forecasts for test set.}
	\label{fig11}
\end{figure}
\section{Conclusion}

This paper applies the softplus function as a link function to the NB-INGARCH model, extending it to the sp NB- INGARCH model. Through approximation of moments and conditional maximum likelihood estimation, we obtain the estimates of the parameters which are validated through a simulation study and data analysis. Furthermore, we upgrade the model to a neural NB-INGARCH model, by incorporating an ANN architecture. Based on real data analysis, we observe that the neural NB-INGARCH model is effective in handling large count data, which is non-stationary and overdispersed.
	\begin{landscape}
	
	\begin{table}[ht]
		\centering
		\caption{Estimates (Standard errors in parantheses), AIC and BIC of neural INGARCH  models fitted to ED data.}
		\renewcommand{\arraystretch}{1.5}
		\scalebox{0.45}{%
			\begin{tabular}{clccccccccccclllllcc}
				\hline
				\multirow{2}{*}{Model} & \multicolumn{1}{c}{\multirow{2}{*}{Constituent Models}}                                                    & \multicolumn{16}{c}{Parameter}                                                                                                                                                                                                                                                                                                                                                                                   & \multicolumn{1}{l}{\multirow{2}{*}{AIC}} & \multicolumn{1}{l}{\multirow{2}{*}{BIC}} \\ \cline{3-18}
				& \multicolumn{1}{c}{}                                                                                       & 1       & 2                          & 3                          & 4                          & 5                    & 6                    & 7                    & 8                    & 9                    & 10                   & 11                   & \multicolumn{1}{c}{12}     & \multicolumn{1}{c}{13}     & \multicolumn{1}{c}{14}     & \multicolumn{1}{c}{15}     & \multicolumn{1}{c}{16}     & \multicolumn{1}{l}{}                     & \multicolumn{1}{l}{}                     \\ \hline
				\multirow{4}{*}{I}     & neu PINGARCH(1,0) $(L=2)$                                                                                  & 28.90   & -22.02                     & 1.63                       & 0.01                       & 2.36                 & -0.19                & \multicolumn{1}{l}{} & \multicolumn{1}{l}{} & \multicolumn{1}{l}{} & \multicolumn{1}{l}{} & \multicolumn{1}{l}{} &                            &                            &                            &                            &                            & \multirow{2}{*}{9976.22}                 & \multirow{2}{*}{10008.14}                \\
				& $(u_1^{1},u_2^{1},u_{11}^{0},u_{21}^{0},u_{12}^{0},u_{22}^{0})$                                            & (20.81) & (11.86)                    & (6.63)                     & (0.04)                     & (0.71)               & (0.06)               & \multicolumn{1}{l}{} & \multicolumn{1}{l}{} & \multicolumn{1}{l}{} & \multicolumn{1}{l}{} & \multicolumn{1}{l}{} &                            &                            &                            &                            &                            &                                          &                                          \\ \cline{2-20} 
				& SARIMAX(3,1,1)(2,1,1){[}24{]} with additional exogenous variables $X_{t-25}$ and $X_{t-30}$                & -0.51   & -0.19                      & -0.10                      & -0.99                      & -0.03                & -0.01                & -1.00                & 0.27                 & 0.01                 & -0.07                & 0.09                 & \multicolumn{1}{c}{0.05}   &                            &                            &                            &                            & \multirow{2}{*}{4690.88}                 & \multirow{2}{*}{4759.83}                 \\
				&                                                                                                            & (0.02)  & (0.03)                     & (0.03)                     & (0.01)                     & (0.03)               & (0.03)               & (0.02)               & (0.05)               & (0.05)               & (0.07)               & (0.02)               & \multicolumn{1}{c}{(0.02)} &                            &                            &                            &                            &                                          &                                          \\ \hline
				\multirow{4}{*}{II}    & neu PINGARCH(2,0) $(L=3)$                                                                                  & 15.00   & \multicolumn{1}{l}{16.67}  & \multicolumn{1}{l}{6.46}   & \multicolumn{1}{l}{-2.19}  & -0.01                & 0.10                 & -2.69                & 0.18                 & 0.03                 & -0.61                & 0.32                 & \multicolumn{1}{c}{-0.48}  &                            &                            &                            &                            & \multirow{2}{*}{9821.70}                 & \multirow{2}{*}{9885.55}                 \\
				& $(u_1^1, u_2^1, u_{11}^0, u_{21}^0, u_{31}^0, u_{12}^0, u_{22}^0, u_{32}^0, u_{13}^0, u_{23}^0, u_{33}^0)$ & (2.67)  & \multicolumn{1}{l}{(3.41)} & \multicolumn{1}{l}{(2.92)} & \multicolumn{1}{l}{(0.28)} & (0.02)               & (0.01)               & (0.28)               & (0.03)               & (0.02)               & (0.67)               & (0.15)               & \multicolumn{1}{c}{(0.23)} &                            &                            &                            &                            &                                          &                                          \\ \cline{2-20} 
				& SARIMAX(4,0,4)(2,1,1){[}24{]} with additional exogenous variables $X_{t-6}$ and $X_{t-25}$                 & -0.50   & 0.12                       & 0.60                       & -0.01                      & 0.05                 & -0.39                & -0.63                & 0.28                 & -0.03                & -0.08                & -0.89                & \multicolumn{1}{c}{0.15}   & \multicolumn{1}{c}{0.02}   & \multicolumn{1}{c}{-0.08}  & \multicolumn{1}{c}{0.09}   & \multicolumn{1}{c}{0.08}   & \multirow{2}{*}{4708.99}                 & \multirow{2}{*}{4794.13}                 \\
				&                                                                                                            & (0.19)  & (0.16)                     & (0.07)                     & (0.07)                     & (0.19)               & (0.09)               & (0.05)               & (0.12)               & (0.03)               & (0.03)               & (0.02)               & \multicolumn{1}{c}{(0.04)} & \multicolumn{1}{c}{(0.04)} & \multicolumn{1}{c}{(0.05)} & \multicolumn{1}{c}{(0.03)} & \multicolumn{1}{c}{(0.03)} &                                          &                                          \\ \hline
				\multirow{4}{*}{III}   & neu PINGARCH(1,1) ($L=1$)                                                                                  & 26.97   & -1.93                      & 0.13                       & 0.02                       & \multicolumn{1}{l}{} & \multicolumn{1}{l}{} & \multicolumn{1}{l}{} & \multicolumn{1}{l}{} & \multicolumn{1}{l}{} & \multicolumn{1}{l}{} & \multicolumn{1}{l}{} &                            &                            &                            &                            &                            & \multirow{2}{*}{9944.00}                 & \multirow{2}{*}{9965.29}                 \\
				& $(u_1^1,u_{11}^0\,u_{21}^0,u_{31}^0)$                                                                      & (0.56)  & (0.03)                     & (0.01)                     & (0.00)                     & \multicolumn{1}{l}{} & \multicolumn{1}{l}{} & \multicolumn{1}{l}{} & \multicolumn{1}{l}{} & \multicolumn{1}{l}{} & \multicolumn{1}{l}{} & \multicolumn{1}{l}{} &                            &                            &                            &                            &                            &                                          &                                          \\ \cline{2-20} 
				& SARIMAX(2,0,1)(2,1,2){[}24{]} with additional exogenous variables $X_{t-25}$, $X_{t-27}$ and $X_{t-30}$    & 0.12    & 0.09                       & -0.62                      & -0.59                      & -0.02                & -0.33                & -0.55                & 0.25                 & 0.02                 & -0.08                & 0.07                 & \multicolumn{1}{c}{0.05}   & \multicolumn{1}{c}{0.06}   & \multicolumn{1}{c}{}       &                            &                            & \multirow{2}{*}{4655.68}                 & \multirow{2}{*}{4724.85}                 \\
				&                                                                                                            & (0.08)  & (0.05)                     & (0.08)                     & (0.08)                     & (0.03)               & (0.08)               & (0.07)               & (0.04)               & (0.04)               & (0.06)               & (0.03)               & \multicolumn{1}{c}{(0.03)} & \multicolumn{1}{c}{(0.03)} & \multicolumn{1}{c}{}       &                            &                            &                                          &                                          \\ \hline
				\multirow{4}{*}{IV}    & neu NB-INGARCH(1,0) ($L=3$)                                                                                & 3.01    & 0.13                       & 15.43                      & 5.67                       & -15.26               & 1.29                 & 0.57                 & -0.46                & 152.76               & -11.18               & \multicolumn{1}{l}{} &                            &                            &                            &                            &                            & \multirow{2}{*}{8351.70}                 & \multirow{2}{*}{9101.97}                 \\
				& ($n, u_1^1, u_2^1, u_3^1, u_{11}^0, u_{21}^0, u_{12}^0, u_{22}^0, u_{13}^0, u_{23}^0$)                     & (9.03)  & (2.03)                     & (0.10)                     & (0.10)                     & (0.08)               & (0.03)               & (0.09)               & (0.10)               & (0.10)               & (0.10)               & \multicolumn{1}{l}{} &                            &                            &                            &                            &                            &                                          &                                          \\ \cline{2-20} 
				& SARIMAX(2,0,0)(2,1,1){[}24{]} with additional exogenous variables $X_{t-18}$ and $X_{t-25}$                & 0.12    & 0.15                       & -0.09                      & -0.10                      & -0.84                & 0.01                 & 0.02                 & -0.01                & 0.06                 & 0.05                 & \multicolumn{1}{l}{} &                            &                            &                            &                            &                            & \multirow{2}{*}{3728.47}                 & \multirow{2}{*}{3875.27}                 \\
				&                                                                                                            & (0.03)  & (0.03)                     & (0.03)                     & (0.02)                     & (0.02)               & (0.00)               & (0.00)               & (0.00)               & (0.03)               & (0.03)               & \multicolumn{1}{l}{} &                            &                            &                            &                            &                            &                                          &                                          \\ \hline
				\multirow{4}{*}{V}     & neu NB-INGARCH(2,0) ($L=1$)                                                                                & 3.02    & 15.43                      & -1.34                      & 11.54                      & 1.32                 & \multicolumn{1}{l}{} & \multicolumn{1}{l}{} & \multicolumn{1}{l}{} & \multicolumn{1}{l}{} & \multicolumn{1}{l}{} & \multicolumn{1}{l}{} &                            &                            &                            &                            &                            & \multirow{2}{*}{7391.51}                 & \multirow{2}{*}{8442.65}                 \\
				& ($n, u_1^1, u_{11}^0, u_{21}^0, u_{31}^0$)                                                                 & (2.20)  & (1.16)                     & (0.08)                     & (0.04)                     & (0.04)               & \multicolumn{1}{l}{} & \multicolumn{1}{l}{} & \multicolumn{1}{l}{} & \multicolumn{1}{l}{} & \multicolumn{1}{l}{} & \multicolumn{1}{l}{} &                            &                            &                            &                            &                            &                                          &                                          \\ \cline{2-20} 
				& SARIMAX(4,0,0)(2,1,1){[}24{]} with additional exogenous variables $X_{t-18}$ and $X_{t-25}$                & 0.12    & 0.15                       & 0.02                       & 0.02                       & -0.07                & -0.10                & -0.85                & 0.01                 & 0.00                 & -0.01                & 0.06                 & \multicolumn{1}{c}{0.03}   & \multicolumn{1}{c}{0.06}   &                            &                            &                            & \multirow{2}{*}{3885.14}                 & \multirow{2}{*}{3900.02}                 \\
				&                                                                                                            & (0.03)  & (0.02)                     & (0.02)                     & (0.03)                     & (0.03)               & (0.02)               & (0.01)               & (0.00)               & (0.00)               & (0.00)               & (0.03)               & \multicolumn{1}{c}{(0.03)} & \multicolumn{1}{c}{(0.03)} &                            &                            &                            &                                          &                                          \\ \hline
				\multirow{4}{*}{VI}    & neu NB-INGARCH(1,1) ($L=1$)                                                                                & 14.90   & 11.69                      & 14.57                      & 8.44                       & 0.73                 & \multicolumn{1}{l}{} & \multicolumn{1}{l}{} & \multicolumn{1}{l}{} & \multicolumn{1}{l}{} & \multicolumn{1}{l}{} & \multicolumn{1}{l}{} &                            &                            &                            &                            &                            & \multirow{2}{*}{7245.16}                 & \multirow{2}{*}{7114.30}                 \\
				& ($n, u_1^1, u_{11}^0, u_{21}^0, u_{31}^0$)                                                                 & (0.01)  & (0.84)                     & (0.03)                     & (0.11)                     & (0.05)               & \multicolumn{1}{l}{} & \multicolumn{1}{l}{} & \multicolumn{1}{l}{} & \multicolumn{1}{l}{} & \multicolumn{1}{l}{} & \multicolumn{1}{l}{} &                            &                            &                            &                            &                            &                                          &                                          \\ \cline{2-20} 
				& SARIMAX(2,0,0)(2,1,1){[}24{]} with additional exogenous variable $X_{t-25}$                                & 0.52    & 0.04                       & -0.33                      & -0.63                      & -0.07                & -0.32                & -0.53                & 4.46                 & 1.34                 & -2.08                & 0.08                 &                            &                            &                            &                            &                            & \multirow{2}{*}{3662.17}                 & \multirow{2}{*}{3756.09}                 \\
				&                                                                                                            & (0.19)  & (0.06)                     & (0.09)                     & (0.08)                     & (0.03)               & (0.06)               & (0.06)               & (0.02)               & (0.08)               & (0.05)               & (0.02)               &                            &                            &                            &                            &                            &                                          &                                          \\ \hline
		\end{tabular}}
		\label{tab6}
	\end{table}
	\begin{table}[H]
		\centering
		\caption{RMSE of out-of-sample forecasts for better performing neural INGARCH  models fitted to ED data.}
		\renewcommand{\arraystretch}{1.2}
		\scalebox{0.8}{%
			\begin{tabular}{lc}
				\hline
				\textbf{Model}                                          & \textbf{RMSE} \\ \hline
				neu - NB-INGARCH(1,1) and SARIMAX(2,0,0)(2,1,1){[}24{]} & 11.93         \\
				neu - NB-INGARCH(2,0) and SARIMAX(4,0,0)(2,1,1){[}24{]} & 13.55         \\
				neu - NB-INGARCH(1,0) and SARIMAX(2,0,0)(2,1,1){[}24{]} & 15.89         \\
				neu - PINGARCH(2,0) and SARIMAX(4,0,4)(2,1,1){[}24{]}   & 22.67         \\ \hline
		\end{tabular}}
		\label{tab7}
	\end{table}
\end{landscape} 
\section*{Acknowledgements}
The first author is grateful to the Cochin University of Science and Technology for the financial support.
\bibliography{ref}
\bibliographystyle{apalike}
\begin{appendices}
	\section{Appendix}
	\subsection{Some useful results from \cite{doukhan2019absolute}}\label{App2}
	
	\noindent \cite{doukhan2019absolute} assume that the process $\{X_t\}$, which is defined on some probability space \((\Omega, \mathcal{F}, P)\), satisfy the equations
	\begin{align}
		X_t \mid \mathcal{F}_{t-1} &\sim D(\lambda_t),  \\ \nonumber
		\lambda_t &= f(X_{t-1}, \dots, X_{t-p}; \lambda_{t-1}, \dots, \lambda_{t-q}), 
	\end{align}
	where \(\mathcal{F}_{t-1}\)  is the $\sigma$-field generated by \(X_{t-1},  X_{t-2},  \dots)\) and \(\{D(\lambda) : \lambda \in [0, \infty)\}\) is some family of univariate distributions. 
	The following definition of a distance metric helps better understand assumption $\mathcal{A}3$ and \ref{App3}.
	\begin{defn}
		The total variation distance between probability measures \( P_1 \) and \( P_2 \), denoted by  \( \text{TV}(P_1, P_2)\) is defined as
		\( \text{TV}(P_1, P_2) = \sup_{A \in \mathcal{B}} |P_1(A) - P_2(A)| \), where \(\mathcal{B} \) is a class of Borel sets.
	\end{defn}
	
	Considering $Y_t = \{X_{t-1},\ldots,X_{t-p+1},\lambda_t,\ldots,\lambda_{t-q+1}\}$, \cite{doukhan2019absolute} impose the following assumptions:
	
	\begin{itemize}
		\item[($\mathcal{A}1$)] There exist positive constants \\\( a_1, \ldots, a_{p-1}, b_0, \ldots, b_{q-1}, \kappa < 1 \), and \( a_0 < \infty \) such that, for
		\begin{equation}
			V(x_1, \ldots, x_{p-1}; \lambda_0, \ldots, \lambda_{q-1}) = \sum_{i=1}^{p-1} a_i x_i + \sum_{j=0}^{q-1} b_j \lambda_j, \nonumber
		\end{equation}
		the condition
		\[
		\mathbb{E}(V(Y_t) \mid Y_{t-1}) \leq \kappa V(Y_{t-1}) + a_0
		\]
		is fulfilled with probability 1.
		
		\item[($\mathcal{A}2$)]  The function \( f \) is measurable and there exist nonnegative constants \( c_1, \ldots, c_q \) with \( c_1 + \cdots + c_q < 1 \) such that
		\[
		|f(x_1, \ldots, x_{p}; \lambda_1, \ldots, \lambda_q) - f(x_1, \ldots, x_{p}; \lambda_1', \ldots, \lambda_q')| \leq \sum_{i=1}^q c_i |\lambda_i - \lambda_i'|
		\]
		for all \( x_1, \ldots, x_p \in \mathbb{R}, \lambda_1, \ldots, \lambda_q, \lambda_1', \ldots, \lambda_q' \geq 0 \).
		
		\item[($\mathcal{A}3$)] There exists some constant \( \delta \in (0, \infty) \) such that
		\[
		\text{TV}(D((\lambda), (\lambda')) \leq 1 - e^{-\delta |\lambda - \lambda'|}
		\]
		for all \( \lambda, \lambda' \geq 0 \).
	\end{itemize}
	\begin{rem}
		For  \( p = q = 1 \), \( Y_t \) reduces to \( \lambda_t \). Condition ($\mathcal{A}1$) follows from the following drift condition 
		
		\begin{itemize}
			\item[($\mathcal{A}1'$)] There exist constants \( \tilde{a}_0 \in [0, \infty), \tilde{a}_1, \ldots, \tilde{a}_p, \tilde{b}_1, \ldots, \tilde{b}_q \in [0, 1] \) with \( \sum_{i=1}^p \tilde{a}_i + \sum_{j=1}^q \tilde{b}_j < 1 \) such that
			\[
			\lambda_t \leq \tilde{a}_0 + \tilde{a}_1 X_{t-1} + \cdots + \tilde{a}_p X_{t-p} + \tilde{b}_1 \lambda_{t-1} + \cdots + \tilde{b}_q \lambda_{t-q}.
			\]
		\end{itemize}
	\end{rem}
	Using a concept of coupling  proposed by \cite{doukhan2019absolute},  \(\lambda_t\) can be expressed as
	\[
	\lambda_t = g(X_{t-1}, X_{t-2}, \dots)
	\]
	for some measurable function \(g\). This yields ergodicity of the process \(\{\lambda_t\}_{t \in \mathbb{Z}}\) and also of the bivariate process \(\{(X_t, \lambda_t)\}_{t \in \mathbb{Z}}\) as stated in the following theorem.
	
	\begin{thm}
		\label{thma2}
		Suppose that the conditions  \((\mathcal{A}1) \rightarrow (\mathcal{A}3)\) are satisfied. Then a stationary version of the process \(\{(X_t, \lambda_t)\}_{t \in \mathbb{Z}}\) is ergodic.
	\end{thm}
	For a proof of the theorem, one can refer \cite{doukhan2019absolute}.
	
	\subsection{Lemma 7 of \cite{gorgi2020beta}}\label{App3}
	\begin{lem}
		The total variation distance between $\mathrm{BNB}_{\lambda_1}$ and $\mathrm{BNB}_{\lambda_2}$ satisfies the inequality
		\[
		\mathrm{TV}(\mathrm{BNB}_{\lambda_1}, \mathrm{BNB}_{\lambda_2}) \leq 1 - \exp(-|\lambda_1 - \lambda_2|)
		\]
		for any $\lambda_1, \lambda_2 \in \mathbb{R}^+$ and any $(n, \alpha) \in \mathbb{R}^+ \times (1, \infty)$. $\mathrm{BNB}_{\lambda_i}$, $i=1,2,$ denote probability measures corresponding to Beta-negative binomial distributions with means $\lambda_1$ and $\lambda_2$ respectively.
	\end{lem}
	\begin{proof}
		Let \( P_\mu \) denote the probability measure of a Poisson distribution with mean \( \mu \) (or $Poi(\mu)$). As demonstrated by \cite{adell2005sharp}, the total variation distance between \( P_{\mu_1} \) and \( P_{\mu_2} \)  bounded as 
		\begin{equation}
			\label{s1}
			\text{TV}(P_{\mu_1}, P_{\mu_2}) \leq 1 - \exp(-|\mu_1 - \mu_2|). 
		\end{equation}
		
		Secondly, we know that if a random variable \( X \) conditional on \( \mu \) has a Poisson distribution with mean \( \mu \), \big(or equivalently $X \big| \mu \sim Poi(\mu)$\big) where \( \mu = \frac{(1 - p) }{p} \mathfrak{z}\), \( p \in (0, 1) \) and \( \mathfrak{z} \) has a gamma distribution with shape parameter \( n \) and scale parameter equal to 1, then \( X \) has a negative binomial distribution, i.e., \( X \sim \text{NB}(n, p) \), with probability mass function (pmf) given by
		\[
		P(X = x) = \frac{\Gamma(x + n)}{\Gamma(x + 1) \Gamma(n)} p^n(1 - p)^x, x=0,1,2,\ldots,\; 0<p \leq 1,\; n >0.
		\]
		
		Next, we use the representation of the negative binomial as a Poisson-gamma mixture to derive an upper bound for the total variation distance between \( \text{NB}_{p_1} \) and \( \text{NB}_{p_2} \), \( p_1, p_2 \in (0, 1) \), where \( \text{NB}_p \) denotes the probability measure of a negative binomial distribution \( \text{NB}(n, p) \). In particular, we set \( \mu_1(\mathfrak{z}) = \frac{(1 - p_1)}{p_1} \mathfrak{z} \) and \( \mu_2(\mathfrak{z}) = \frac{(1 - p_2)}{p_2} \mathfrak{z} \), and $P_{\mu_1(\mathfrak{z})}$ and $P_{\mu_2(\mathfrak{z})}$ as the probability measures corresponding to $Poi(\mu_1)$ and $Poi(\mu_2)$ respectively. Hence, we
		obtain  \( \text{NB}_{p_1}(A) = E[ P_{\mu_1(\mathfrak{z})}(A) ]\) and \( \text{NB}_{p_2}(A) = E[ P_{\mu_2(\mathfrak{z})}(A)] \) for any \( A \subseteq \mathbb{N} \). Therefore, from the expression of the total variation distance and the upper bound in (\ref{s1}), we obtain
		\begin{align}
			\label{a3}
			\text{TV}(\text{NB}_{p_1}, \text{NB}_{p_2}) &\leq E[ \text{TV}(P_{\mu_1(\mathfrak{z})}, P_{\mu_2(\mathfrak{z})})]\\ \nonumber
			&\leq 1 - \exp \left( - E[ |\mu_1(\mathfrak{z}) - \mu_2(\mathfrak{z})|] \right)\\ \nonumber
			&\leq 1 - \exp \left( - n \left| \frac{1 - p_1}{p_1} - \frac{1 - p_2}{p_2} \right| \right). 
		\end{align}
		Note that we only require the proof till (\ref{a3}) to arrive at (\ref{a3org}). The rest of the proof relies on the representation of the Beta - negative binomial (BNB) as a beta mixture of negative binomials to obtain the desired upper bound for \( \text{TV}(\text{BNB}_{\lambda_1}, \text{BNB}_{\lambda_2}) \). 
	\end{proof}
	\subsection{Model diagnostic analysis for Syphilis counts data}\label{diag}
	\begin{figure}[H]
		\centering
		\includegraphics[scale=0.6]{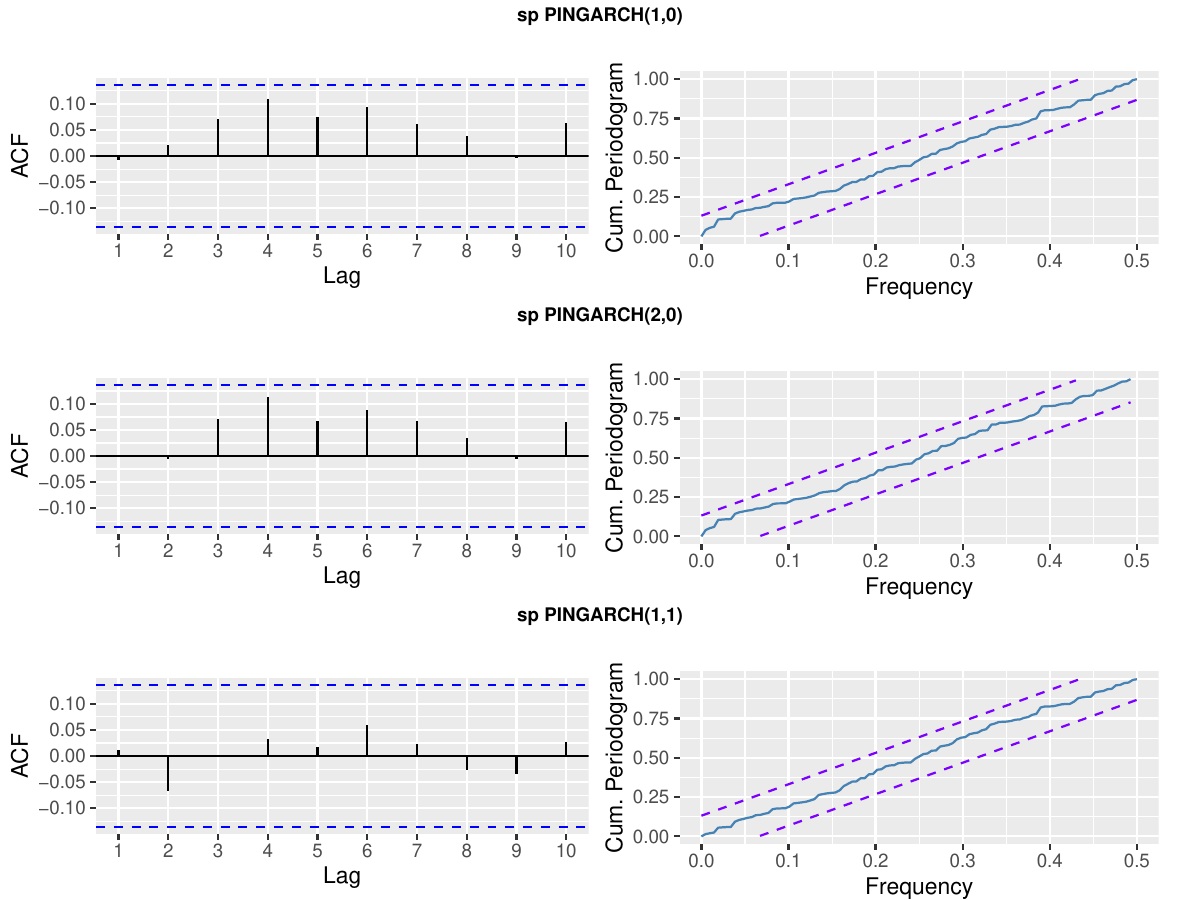}
		\caption{ACF plots and cumulative periodograms for residuals of softplus PINGARCH models fitted to Syphilis data }
		\label{fig7}
	\end{figure}
	
	\begin{figure}[H]
		\centering
		\includegraphics[scale=0.66]{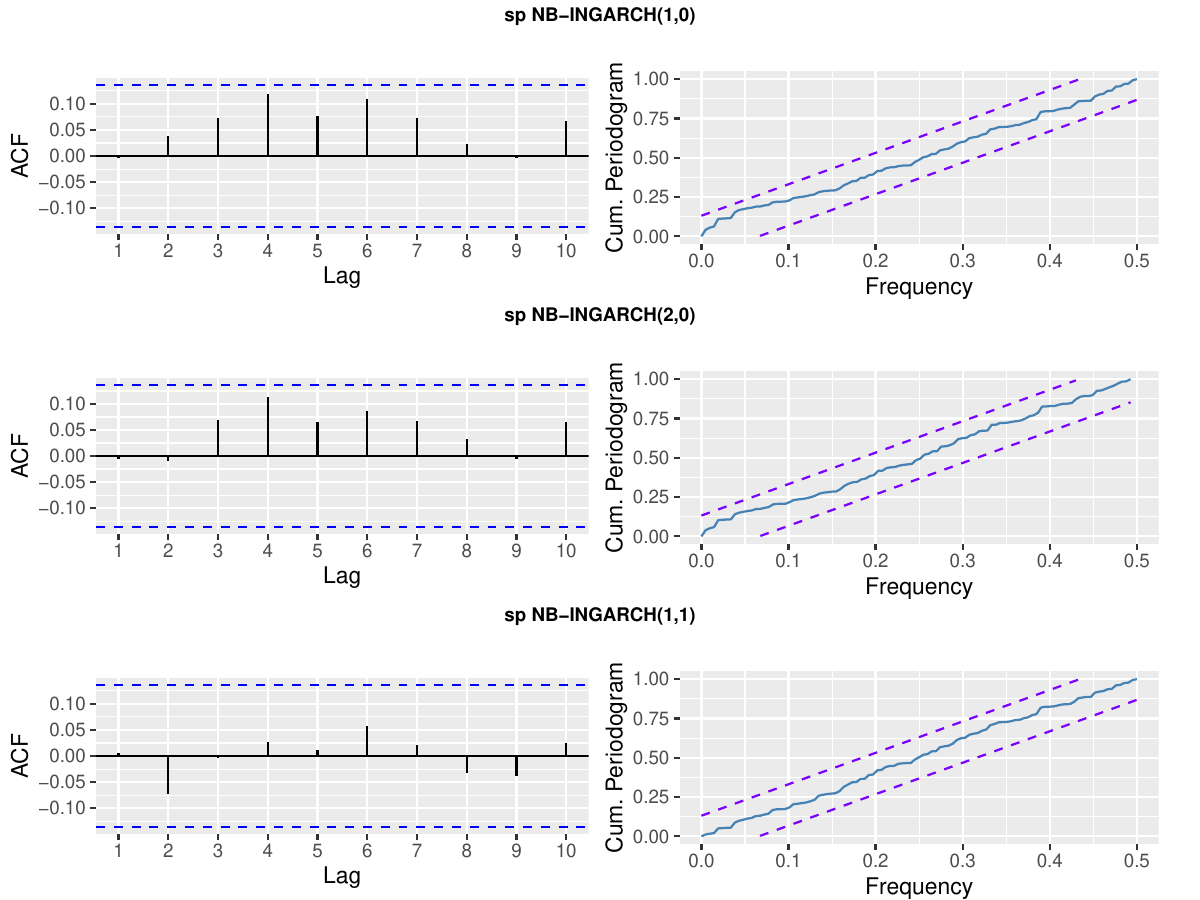}
		\caption{ACF plots and cumulative periodograms for residuals of softplus NB-INGARCH models fitted to Syphilis data }
		\label{fig8}
	\end{figure}
	
	\subsection{Additional plots for analysis of ED data}\label{explore}
	\begin{figure}[H]
		\centering
		\includegraphics[scale=0.66]{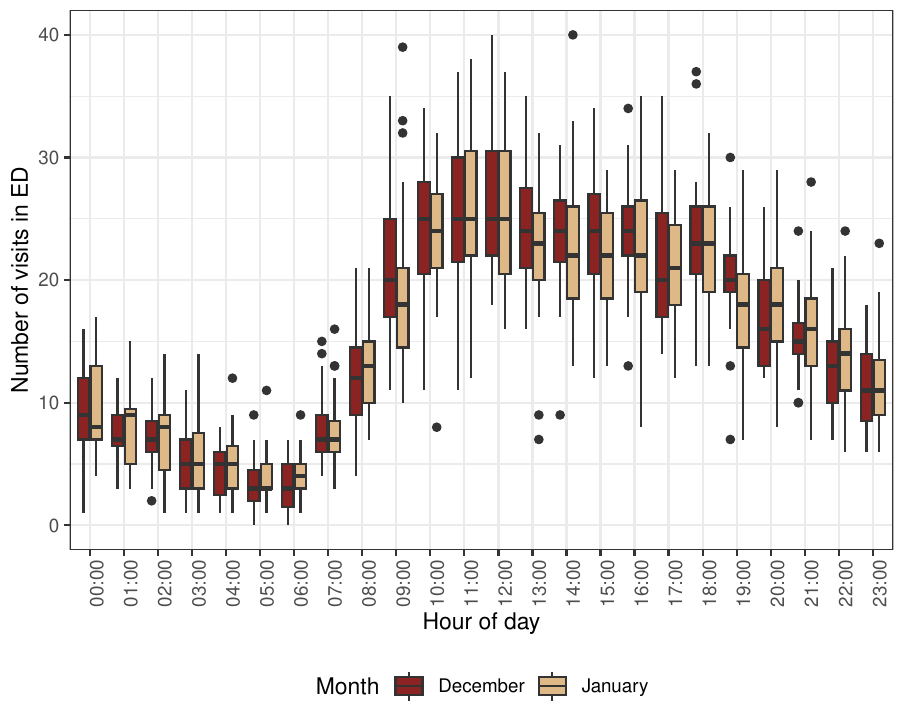}
		\caption{Boxplot of ED arrivals based on hours of a day (in 24 - hour format).}
		\label{fig2}
	\end{figure}
	\begin{figure}[H]
		\centering
		\includegraphics[scale=0.5]{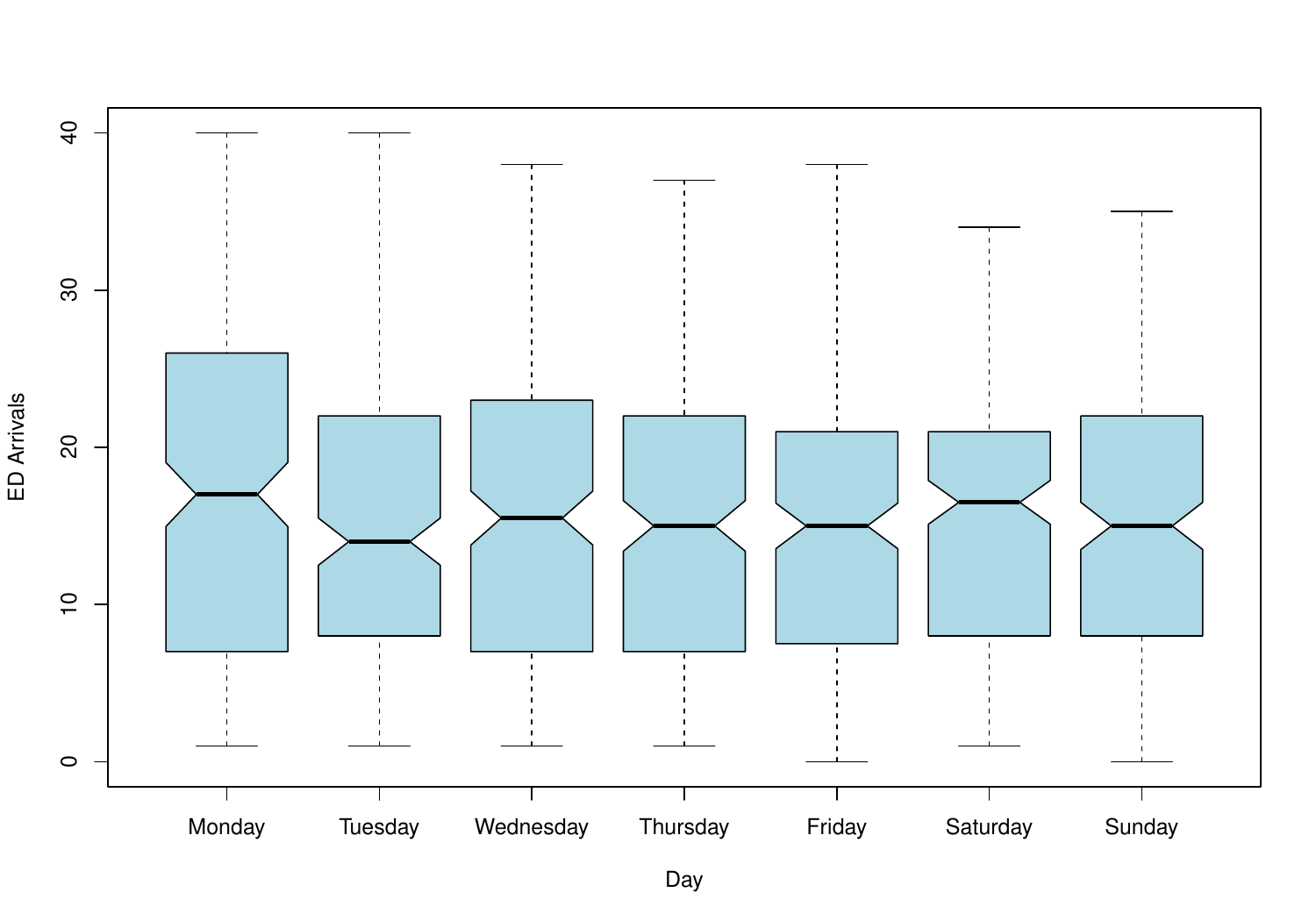}
		\caption{Boxplot of ED arrivals based on days of the week.}
		\label{fig3}
	\end{figure}
	\begin{figure}[H]
		\centering
		\includegraphics[scale=0.6]{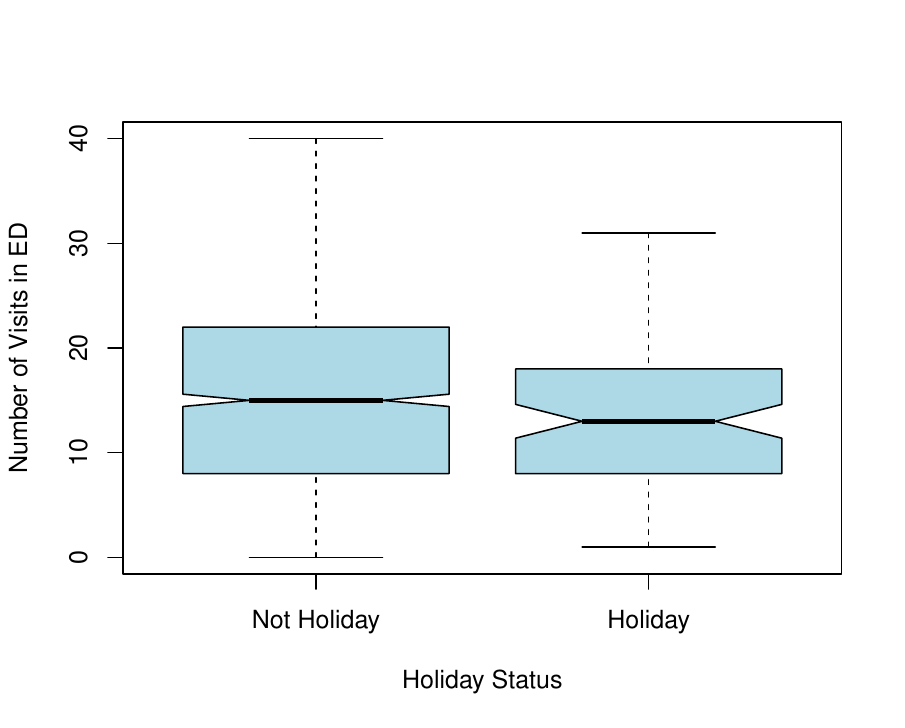}
		\caption{Boxplot of ED arrivals based on school and bank holidays.}
		\label{fig4}
	\end{figure}
	\begin{figure}[H]
		\centering
		\includegraphics[scale=0.66]{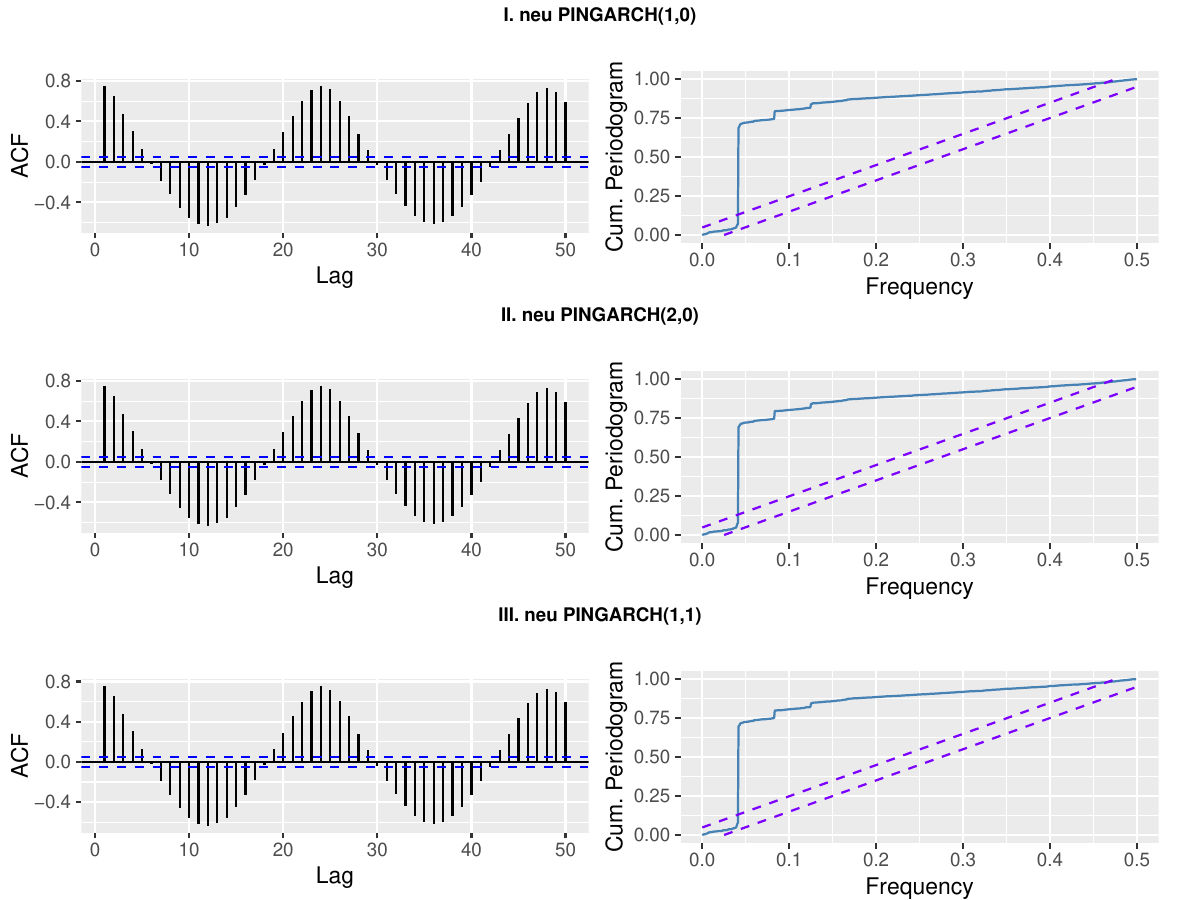}
		\caption{ACF plots and cumulative periodograms for residuals of neu PINGARCH fits to the data}
		\label{fig5}
	\end{figure}
	\begin{figure}[H]
		\centering
		\includegraphics[scale=0.6]{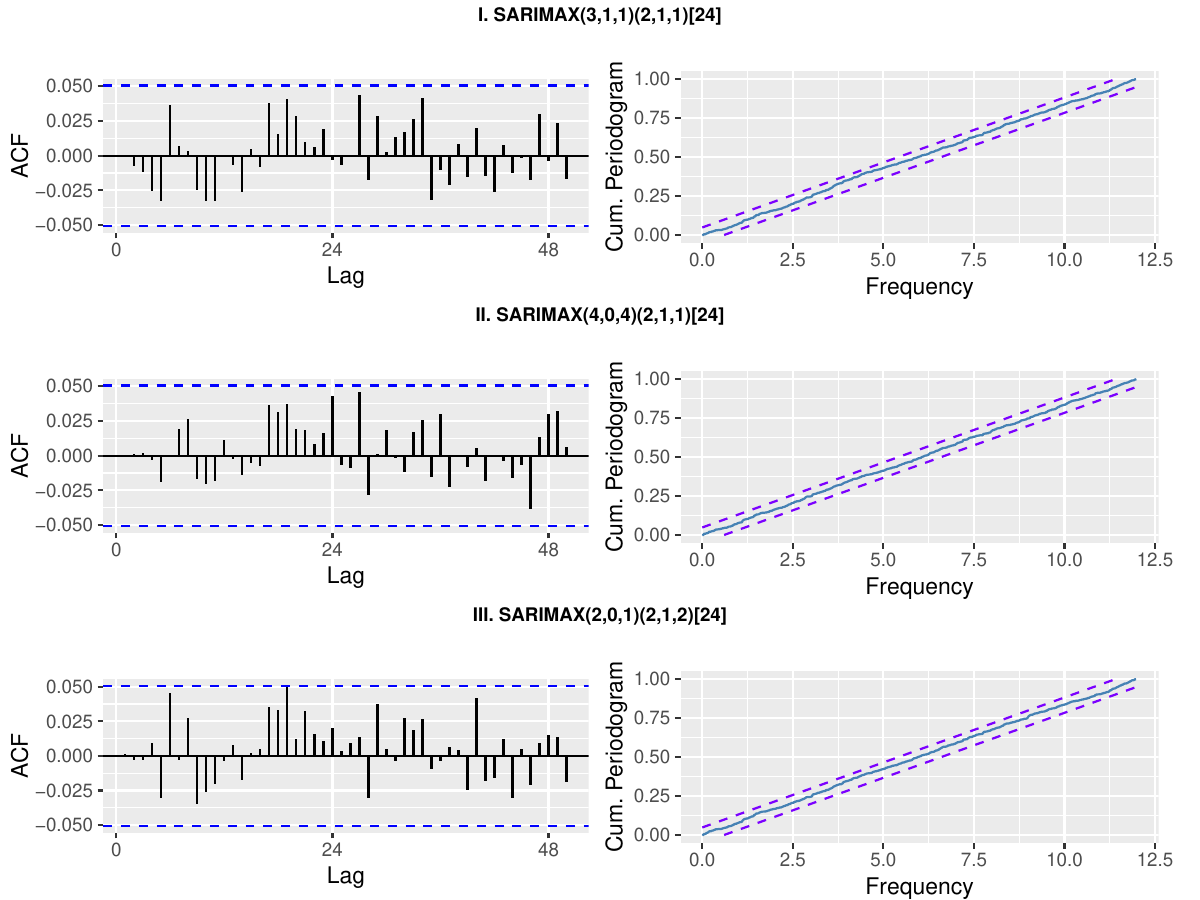}
		\caption{ACF plots and cumulative periodograms for residuals of SARIMAX fits to the residuals of neural Poisson fits. }
		\label{fig6}
	\end{figure}
	\begin{figure}[H]
		\centering
		\includegraphics[scale=0.6]{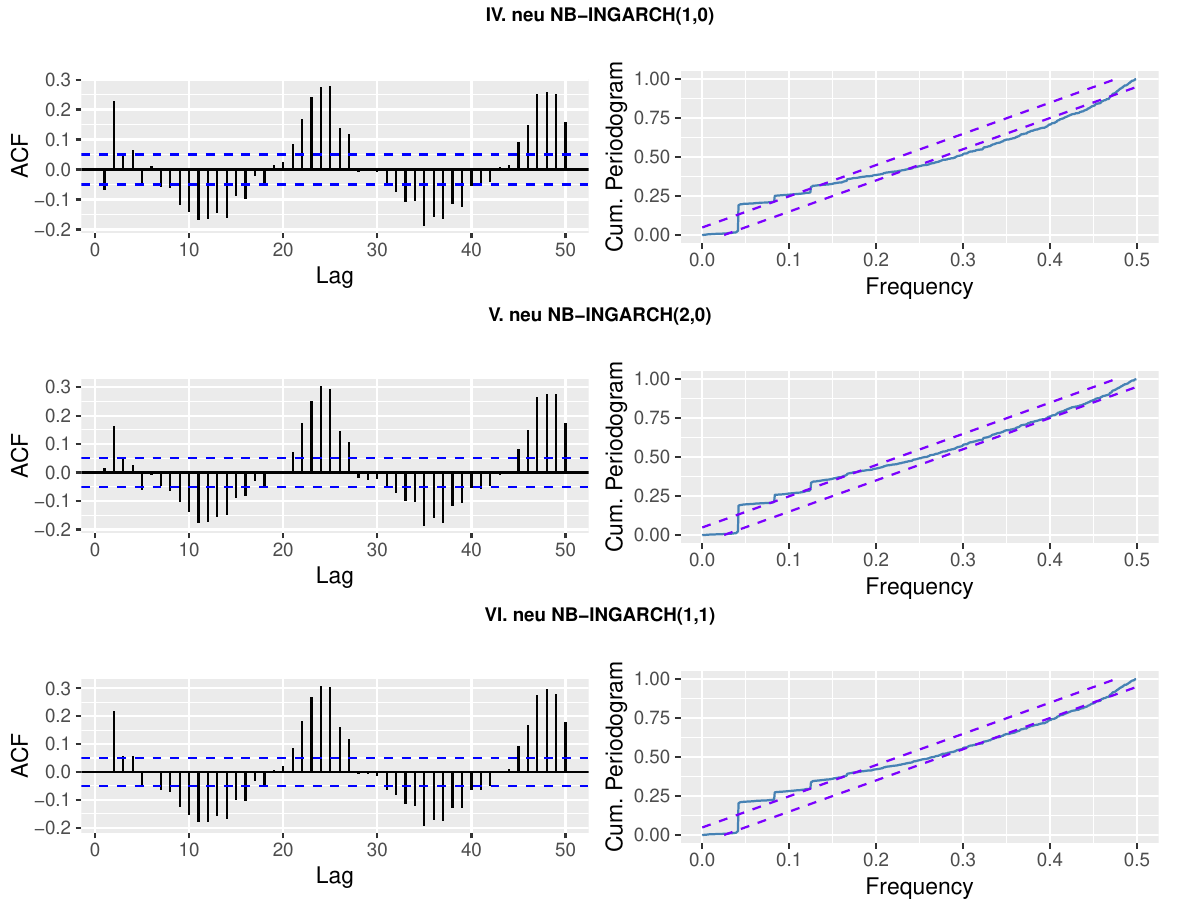}
		\caption{ACF plots and cumulative periodograms for residuals of neu NB-INGARCH fits to the data}
		\label{fig9}
	\end{figure}
	\begin{figure}[H]
		\centering
		\includegraphics[scale=0.66]{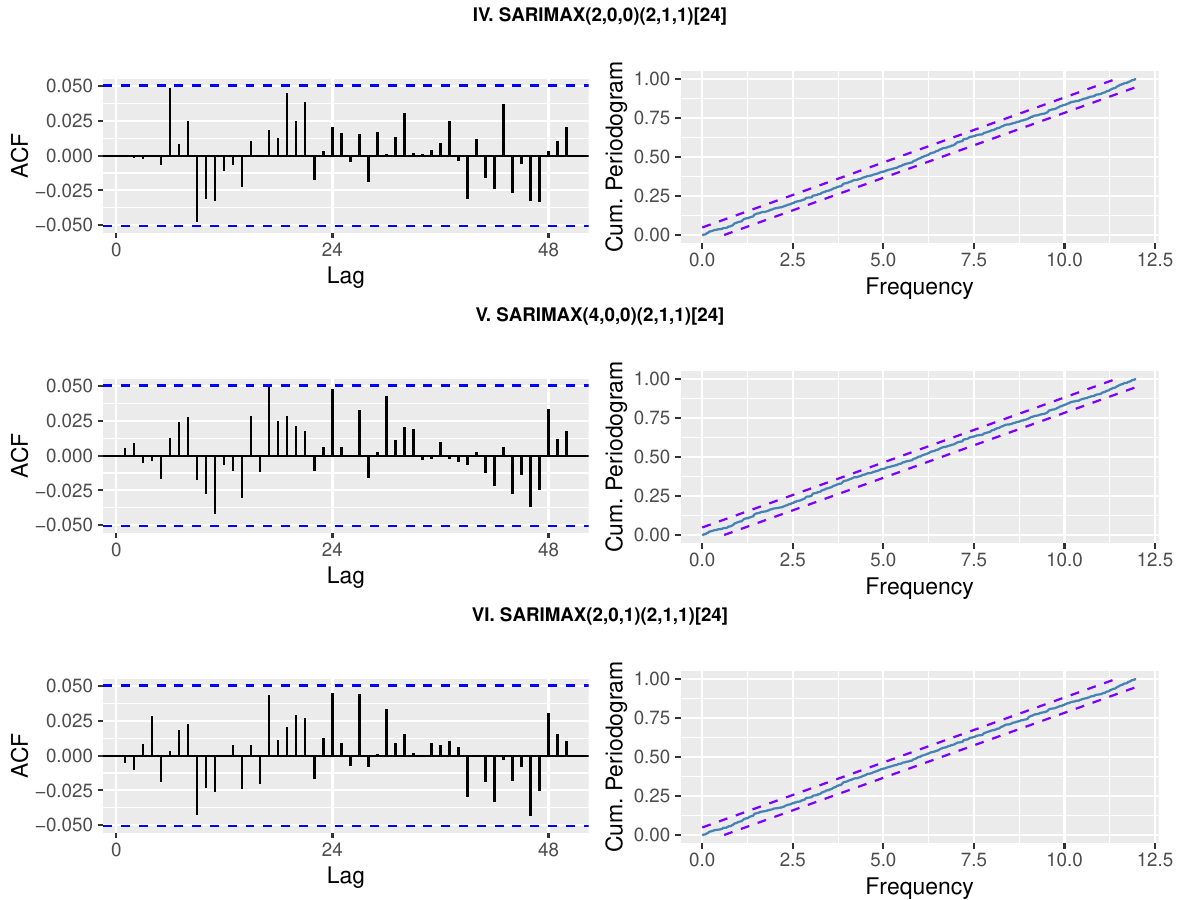}
		\caption{ACF plots and cumulative periodograms for residuals of SARIMAX fits to the residuals of neural NB-INGARCH fits.}
		\label{fig10}
	\end{figure}
\end{appendices}
\end{document}